\documentclass[12pt,letterpaper]{article}
\usepackage[utf8]{inputenc}
\usepackage[top=1in, bottom=1in, left=1in, right=1in]{geometry}
\usepackage{amsmath}
\usepackage{amsthm}
\usepackage{amssymb}
\usepackage{graphicx}
\usepackage{hyperref}
\usepackage{tabu}
\usepackage{longtable}
\usepackage{booktabs}
\usepackage{dsfont}
\usepackage{tabularx}
\usepackage{fancyhdr}
\usepackage{enumitem}
\usepackage{multirow}
\usepackage{multicol}
\usepackage[bottom]{footmisc}
\usepackage{bm}
\usepackage{float}
\usepackage[english]{babel}
\usepackage{color}
\usepackage{verbatim}
\usepackage{mathtools}
\usepackage{esvect}
\usepackage{mathabx}
\usepackage{blkarray}
\usepackage{blindtext}
\usepackage{algorithm} 
\usepackage{algorithmic} 
\usepackage{xcolor}
\usepackage{natbib}
\usepackage{setspace}
\usepackage[toc,page]{appendix}
\usepackage{siunitx}
\usepackage{authblk}
\usepackage{soul}
\newtheorem{theorem}{Theorem}
\newtheorem{assumption}{Assumption}
\newtheorem{lemma}{Lemma}

\newtheorem{proposition}{Proposition}

\newcommand{\Var}{\text{Var}}

\newcommand{\real}{\mathbb{R}}

\newcommand{\vecx}{\mathbf{x}}

\newcommand{\vecu}{\mathbf{u}}

\setcitestyle{citesep={,}}

\allowdisplaybreaks 

\title{Wilcoxon Random Forests for Robust Distributional Prediction}
\author[1]{Danni Shi\thanks{dshidanni@outlook.com}}
\author[1]{Bryan E. Shepherd\thanks{bryan.shepherd@vanderbilt.edu}}
\author[2]{Chun Li\thanks{cli77199@usc.edu}}

\affil[1]{Department of Biostatistics, Vanderbilt University}
\affil[2]{Department of Population and Public Health Sciences, University of Southern California}
\date{}

\begin{document}
\maketitle
\begin{abstract}
Random forests (RF) are tree-based models that capture complex, high-order predictor interactions. In biomedical studies, ordered outcomes often have outliers, skewness, or mixtures of continuous and ordinal values (e.g., due to detection limits). Standard RF splitting targets conditional means via squared errors and is sensitive to such irregularities. We propose a Wilcoxon regression tree that selects splits by maximizing a rank-based impurity reduction, equivalent to maximizing the squared Wilcoxon rank-sum statistic.
Because the criterion depends only on outcome ranks, it is invariant to monotone transformations and naturally accommodates continuous, ordinal, or mixed outcomes. We develop the Wilcoxon random forest (WRF), which aggregates Wilcoxon regression trees via subsampling and estimates conditional distributions using forest-weighted empirical CDFs. We establish consistency of the WRF distribution estimator under regularity conditions. 
We evaluate distributional prediction using calibration diagnostics and continuous ranked probability scores, and define an out-of-bag permutation variable-importance measure.
Simulations show that the WRF performs comparably to the standard quantile regression forest when errors are symmetric and homoscedastic, and yields improved quantile estimation and calibration when outcomes are skewed or heteroscedastic. We apply the WRF to predict CD4 cell count and HIV viral load six months after antiretroviral therapy initiation in a multicenter Latin American cohort, and the WRF improves threshold probability estimation and accommodates the large mass at the viral load detection limit.
\end{abstract}


\newpage
\section{Introduction}
Random forests (RF) are widely used in biomedical research for prediction, classification and variable selection in settings with complex covariate effects and high-order interactions \citep{breiman2001}. In regression RFs, trees are usually grown by recursive binary partitioning of the covariate space, with splits chosen to reduce within-node sum of squared errors (SSE). This splitting criterion effectively targets conditional mean estimation and can be sensitive to skewness, outliers, and heteroscedasticity, which are common in biomedical studies. Moreover, many biomarkers are recorded on ordinal scales or subject to lower and/or upper detection limits such that the observed outcome is a mixture of continuous and ordered point masses. These settings motivate the need for random forest methods that predict conditional distributions while remaining robust to outcome scales and extreme observations.

Our motivating application is the prediction of treatment response among people living with HIV (PLWH) six months after the initiation of antiretroviral therapy (ART). Consider two important biomarkers for HIV: CD4 cell count and HIV-1 RNA viral load (VL). CD4 cell count is often skewed, and VL is highly skewed and subject to assay detection limits. Prediction of these outcomes requires methods that can estimate conditional distributions for continuous and mixed continuous-ordinal outcomes while remaining robust to skewness and outliers.

Quantile regression forests (QRF; \citealp{meinshausen2006}) extend RFs by using forest-induced nearest-neighbor weights to estimate conditional quantiles and potentially the full conditional distribution. 
The prediction is distributional, but the trees are still grown by SSE splitting and the partitioning remains mean-oriented. 
Quantile regression trees \citep{chaudhuri2002nonparametric} and pinball-loss-based splitting rules \citep{bhat2015} directly target quantiles, but they remain scale-dependent and require the quantile level(s) of interest to be specified before the tree is grown.

In this paper, we propose a Wilcoxon splitting criterion for generating tree structures. With this criterion, a split is chosen to maximize separation between the outcome ranks in the two child nodes: this criterion is equivalent to maximizing the squared Wilcoxon rank-sum statistic among all candidate splits.
Because this criterion depends only on the relative ordering of the outcomes, it is invariant to strictly monotone transformations and naturally accommodates skewed, heavy-tailed, ordinal, and mixed-type outcomes. 
The resulting partition targets differences in the relative ordering of the outcomes rather than differences in raw conditional means, and our splitting criterion does not require prespecifying quantile level(s).

Once we have obtained a tree structure, the Wilcoxon regression tree is defined with the empirical cumulative distribution function (ECDF) per terminal node. 
We then develop the Wilcoxon random forest (WRF), which aggregates Wilcoxon regression trees built on subsamples and estimates the conditional CDF using forest-weighted empirical distributions. Under regularity conditions, we show that the estimated conditional CDF from the WRF is consistent for various outcome types, including continuous, ordinal, and mixed outcomes.

The remainder of the paper is organized as follows. Section 2 reviews splitting criteria based on SSE and pinball loss, introduces our rank-based splitting criterion, defines the Wilcoxon regression tree, and describes an optional semi-parametric refinement for single-tree prediction.
Section 3 presents the WRF algorithm for predicting conditional distributions and describes its properties and various evaluation metrics. Section 4 demonstrates the performance of the WRF through simulations. Section 5 reports the HIV application, and Section 6 concludes with a brief discussion.

\section{Rank-Based Wilcoxon Regression Tree}

\subsection{Review of Splitting Criteria: Squared Errors and Pinball Loss}
In the Classification and Regression Tree (CART) framework \citep{CARTbook}, regression trees are constructed via recursive binary splitting of the covariate space $\mathcal{X} \subseteq \real^p$ using a training sample $\mathcal{D}_n= \{(\mathbf{X}_i, Y_i): i=1, \ldots, n\}$ drawn i.i.d. from the joint distribution of $(\mathbf{X}, Y)$, where $\mathbf{X}=(X_1,\ldots,X_p)^\top$ denotes the predictor vector.
At any node $A \subseteq \mathcal{X}$, a candidate split is defined by a covariate $j \in \{1,\ldots,p\}$ and a threshold $s$, together denoted as $(j,s)$, and it partitions $A$ into two child regions: $A_L(j,s)= \{\vecx\in A: x_j \leq s\}$ and $A_R(j,s) = \{\vecx \in A: x_j > s\}$.
For unordered categorical $X_j$ with $L$ levels, the standard CART implementation orders the levels within the current node by their within-level outcome means, and considers the $L-1$ splits induced by this ordering \citep{CARTbook}. The selected split is the one that maximizes a prespecified impurity reduction.

For regression trees, the classical CART impurity is the sum of squared errors (SSE). If a node $A$ of size $n_A$ is split at $(j,s)$ with child nodes $A_L(j,s)$ and $A_R(j,s)$, the reduction in SSE is 
\begin{equation}
\label{eq:loss_SSE}
\Delta_{\text{SSE}}(j,s)= \sum_{i:\mathbf{X}_i\in A } (Y_i- \bar{Y}_A)^2- \left[ \sum_{i: \mathbf{X}_i \in A_L (j,s)} (Y_i -\bar{Y}_L)^2+ \sum_{i: \mathbf{X}_i \in A_R (j,s)} (Y_i -\bar{Y}_R)^2 \right],
\end{equation}
where $\bar{Y}_A$, $\bar{Y}_L$ and $\bar{Y}_R$ are the outcome means in the parent, left child, and right child nodes, respectively.  The optimal split $(j^*, s^*)$ maximizes $\Delta_{\text{SSE}}(j,s)$. The SSE-based splitting is used in standard CART regression trees, random forests, and quantile regression forests (QRF; \citealp{meinshausen2006}). This criterion is well suited for separating conditional means, but it is scale-dependent and can be sensitive to skewness, outliers and heteroscedasticity.

To directly target conditional quantile estimation, \cite{bhat2015} proposed a splitting criterion based on the pinball loss (also known as the quantile loss or check function). Given a specified quantile level $\tau\in (0,1)$, the pinball loss is defined as
\begin{align*}
  \rho_\tau (y, \hat{q}_\tau)= \begin{cases}
  \tau (y-\hat{q}_\tau) &\text{if } y\geq \hat{q}_\tau, \\
  (1-\tau) (\hat{q}_\tau- y) &\text{if } y< \hat{q}_\tau, \end{cases}
\end{align*}
where $\hat{q}_\tau$ is the predicted $\tau$-th quantile \citep{koenker1978}. The corresponding pinball-loss-based impurity reduction is:
\begin{align*}
\Delta_{\text{pinball}, \tau} (j,s)=  \sum_{i:\mathbf{X}_i\in A} \rho_{\tau} (Y_i, \hat{q}_{\tau A}) - \left[ \sum_{i: \mathbf{X}_i \in A_L(j,s)} \rho_{\tau} (Y_i, \hat{q}_{\tau L})+ \sum_{i: \mathbf{X}_i \in A_R(j,s)} \rho_{\tau} (Y_i, \hat{q}_{\tau R}) \right],
\end{align*}
where $\hat{q}_{\tau A}$, $\hat{q}_{\tau L}$ and $\hat{q}_{\tau R}$ are the empirical $\tau$-th quantiles (from a quantile function) of the outcomes in the parent, left child, and right child nodes, respectively. Multiple quantile levels can be handled by maximizing the sum of $\Delta_{ \text{pinball}, \tau} (j, s)$ over a prespecified set of $\tau$'s. 

While the pinball-loss-based splitting directly targets conditional quantiles and is less mean-oriented than SSE, it relies on the prespecified quantile level(s) and it still depends on the outcome scales. As a result, it may be affected by a few extreme observations. In addition, targeting multiple quantile levels requires repeated empirical-quantile calculations at each candidate split.

\subsection{Wilcoxon Rank-Based Node Splitting}
\label{sec:2_2-wilcoxsplit}
We propose a rank-based splitting criterion that is invariant to the outcome scale and is naturally applicable to continuous, ordinal, or mixed-type outcomes. 
Given a parent node $A$ of size $n_A$, let $r_{i,A}$ denote the midrank of $Y_i$ among the outcomes $\{Y_i: \mathbf{X}_i \in A\}$. For a candidate split $(j,s)$, define
\begin{align*}
    \bar{r}_A =\frac1{n_A}\sum_{i:\mathbf{X}_i\in A} r_{i,A}, \quad
    \bar{r}_L =\frac1{n_L}\sum_{i:\mathbf{X}_i\in A_L(j,s)} r_{i,A}, \quad
    \bar{r}_R =\frac1{n_R}\sum_{i:\mathbf{X}_i\in A_R(j,s)} r_{i,A}, 
\end{align*}
where $n_L$ and $n_R$ are the left and right child node sizes, and $\bar{r}_A$, $\bar{r}_L$ and $\bar{r}_R$ are the average midrank in the parent, left child, and right child nodes, respectively, 
We define the rank-based impurity reduction as
\begin{equation}
\label{eq:loss_rank}
\Delta_{\text{rank}}(j, s)= \sum_{i:\mathbf{X}_i\in A} (r_{i,A}- \bar{r}_A)^2 - \left[ \sum_{i: \mathbf{X}_i \in A_L (j,s)} (r_{i,A} -\bar{r}_{L})^2+ \sum_{i: \mathbf{X}_i \in A_R (j,s)} (r_{i,A} -\bar{r}_{R})^2 \right].
\end{equation}
This has the same form as the SSE-based impurity reduction in (\ref{eq:loss_SSE}), but it is computed on the node-specific midranks rather than on the raw $Y_i$ values. Thus, a split is favored when it produces child nodes with well-separated outcome ranks.

We call $\Delta_{\text{rank}}(j,s)$ the Wilcoxon splitting criterion. This terminology is motivated by the fact that the criterion in (\ref{eq:loss_rank}) is equivalent to a Wilcoxon rank-sum / Mann-Whitney statistic. To formalize this equivalence, for a candidate split $(j,s)$, consider the normalized Mann-Whitney $U$ statistic for comparing the outcomes in the left and right child nodes:
\begin{align*}
U_A(j,s)= \frac{1}{n_L n_R} \sum_{i: \mathbf{X}_i \in A_L (j, s)} \sum_{k: \mathbf{X}_k \in A_R (j, s)} \left[I(Y_i >Y_k)+ \frac{1}{2} I(Y_i =Y_k) \right].
\end{align*}
Under the null hypothesis that the left and right child node outcomes have the same distribution, $U_A(j,s)$ can be normalized as
\begin{align*}
z(j,s)= \frac{U_A(j,s)- \frac{1}{2}}{ \sqrt{\Var (U_A(j,s))}}, \text{ where } \Var[U_A(j,s)]= \frac{n_A+1- \frac{\kappa}{ n_A (n_A-1)}}{ 12 n_L n_R},
\end{align*}
and $\kappa= \sum_{k=1}^K (t_k^3-t_k)$ is the tie-correction factor for $K$ unique tied-rank groups of sizes $t_1, \ldots, t_K$ in the parent node. Since $z(j,s)\xrightarrow{d} N(0,1)$, the two-sided Wilcoxon test $p$-value can be computed as $p_W (j,s)= 2[1- \Phi(|z(j,s)|)]$, where $\Phi(\cdot)$ is the cumulative distribution function (CDF) of the standard normal distribution. 

\begin{proposition}[Equivalence of rank-based splitting and the Wilcoxon rank-sum test]
\label{thm:prop1}
The rank-based impurity reduction (\ref{eq:loss_rank}) can be expressed as 
\begin{equation}
\label{eq:loss_rank_equiv}
\Delta_{\text{rank}} (j,s) = n_A n_L n_R \left[ U_A(j,s)- \frac{1}{2}\right]^2
=\frac{n_A}{12} \left[ n_A+1- \frac{\kappa}{ n_A(n_A-1)} \right] z(j,s)^2.
\end{equation}
\end{proposition}
The proof is in Supplementary~S1.2. Equation (\ref{eq:loss_rank_equiv}) implies that, within a fixed parent node, $\Delta_{\text{rank}} (j,s)$ is a strictly increasing function of $z(j,s)^2$. Since $p_W(j,s)$ is a strictly decreasing function of $z(j,s)^2$, the optimal split $(j^*, s^*)$ satisfies
\begin{align*}
(j^*, s^*)= \arg \max_{j,s} \Delta_{\text{rank}} (j,s) = \arg \max_{j,s} z(j,s)^2 = \arg \min_{j,s} p_W(j,s).
\end{align*}
The $p$-value is used only as an equivalent split score, not as a formal hypothesis test.

For a numerical predictor $X_{j}$, following the standard CART practice \citep{CARTbook, scornetetal2015}, we define the threshold of a split to be the midpoint between consecutive distinct values of $X_{j}$ within the parent node. This convention ensures that no split boundary coincides with an observed value, improving boundary precision for future observations. For an unordered categorical predictor $X_j$ with $L$ levels in node $A$, we order the levels by their within-level mean midranks in $A$ and consider the $L-1$ splits induced by this ordering.

To regularize the tree structure, we can impose standard constraints such as minimum node size, maximum tree depth, maximum number of terminal nodes, and minimum reduction of rank-based impurity \citep{CARTbook}.
A rank-based analog of cost-complexity pruning is described in Supplementary Section~S1.1.

Computationally, for an ordered predictor, searching over candidate Wilcoxon splits has the same $O(n \log n)$ order as the corresponding CART split search, with additional computation required to obtain the node-specific outcome midranks.

\subsection{Wilcoxon Regression Tree and Distributional Prediction}
\label{sec:2_3-pred}

After the recursive binary splitting and the optional pruning, we have a rank-based partition of the covariate space into terminal nodes.  A distributional prediction model conditional on covariates can then be constructed by estimating a predicted outcome distribution for each terminal node of the partition.  A simple approach is to predict the distribution within each node using the ECDF of the observations that fall into the node.  We call the resulting tree model the Wilcoxon regression tree.

The leaf-wise ECDFs in the Wilcoxon regression tree can be highly variable, particularly for small leaves.  To improve efficiency, we can refine it by modeling the leaf-wise conditional distributions using a semiparametric, rank-based model such as the cumulative probability model (CPM). Specifically, for a tree structure $T$, let $\mathcal{L}(T)=\{\ell_1,\ell_2,\ldots,\ell_{M_T}\}$ denote its terminal nodes, where $M_T=|\mathcal{L}(T)|$. Let $I_m(\vecx)=I(\vecx\in \ell_m)$ indicate that $\vecx$ falls in terminal node $\ell_m$ ($m=2,\ldots,M_T$). The conditional CDF can be expressed as
\begin{align*}
G[F(y\mid \vecx)]= \alpha(y)- \sum_{m=2}^{M_T} \beta_m I_m(\vecx) ,
\end{align*}
where $G(\cdot)$ is a specified link function (e.g., logit), $\beta_m$'s are regression coefficients for the node indicators, and $\alpha(y)$ is an intercept function that will be estimated nonparametrically. In this formulation, the link-transformed conditional CDF is assumed to share a common, nonparametrically estimated shape, $\alpha(y)$, but differ via a location shift across nodes.  As a rank-based method, the CPM is closely related to Wilcoxon rank-sum test \citep{mccullagh1980, tian2024}, and can be used for any orderable outcome including continuous, ordinal discrete, or mixed type outcomes \citep{liu2017}. 
Because the CPM borrows information across leaves, the resulting tree model, which we call the CPM regression tree, can give more stable predicted distributions than the fully nonparametric empirical estimates from the Wilcoxon regression tree.

\section{Wilcoxon Random Forest}
\label{sec:3-WRF}

\subsection{Algorithm for Wilcoxon Random Forest}
\label{sec:3_1-WRFalgo}

The Wilcoxon Random Forest (WRF) is an ensemble of rank-based Wilcoxon regression trees described in Section 2.  Our goal is to construct a rank-based random forest approach for estimating the full conditional distribution given covariates.

Let $\{(\mathbf{X}_i, Y_i)\}_{i=1}^n$ denote the training data, where $Y_i$ is the response variable and $\mathbf{X}_i$ is a vector of $p$ covariates.  We first grow $B$ Wilcoxon regression trees, $T_1, \ldots, T_B$, independently.  Each tree structure is grown as follows:
\begin{enumerate}
\item For tree $T_b$, let $\mathcal{I}_b \subset \{1,\ldots,n\}$ denote the indices of a subsample of size $\lceil (1-e^{-1})n\rceil\approx0.632n$, drawn \textit{without} replacement. Grow $T_b$ using the observations indexed by $\mathcal I_b$.

\item At each node, randomly select $\textit{mtry}$ candidate predictors from the $p$ predictors. 

\item For each selected predictor and each candidate threshold, compute $\Delta_{\text{rank}}$ as in (\ref{eq:loss_rank}). 

\item Choose the split $(j^*, s^*)$ that maximizes $\Delta_{\text{rank}} (j,s)$ and partition the node accordingly. 

\item Repeat steps 2--4 until a stopping criterion is met.
\end{enumerate}

Tree growth is regularized by standard stopping criteria such as a minimum terminal node size, a maximum tree depth, or a minimum rank-based impurity reduction.
Here, each tree is grown on a subsample drawn without replacement rather than on a bootstrap sample. 
We use subsampling so that each selected observation contributes only once to the splitting criterion and once to the CDF estimation. Subsampling has been used as an alternative to bootstrap sampling in tree ensembles \citep{BuhlmannYu2002}.
We set the subsample size to $\lceil (1-e^{-1})n \rceil \approx 0.632n$ to match the expected number of distinct observations in a bootstrap sample. Other subsampling fractions are acceptable as long as they contain sufficient observations.

Once the tree structure has been grown, the corresponding tree model for the conditional CDF is just the ECDF for each terminal node.  Specifically, given a new covariate $\vecx$, let $L_b(\vecx)$ denote the leaf in tree $T_b$ that contains $\vecx$, and $N_b(\vecx) = \sum_{i=1}^n I(\mathbf{X}_i \in L_b(\vecx))$ denote the number of training observations falling in that leaf. Write $F(y \mid \vecx) = P(Y \leq y \mid \mathbf{X}= \vecx)$.
The tree-level conditional CDF estimator given $\vecx$ from $T_b$ is
\begin{equation}
\label{eq:tree_ECDF}
\hat{F}_b(y\mid \vecx) =\frac{1}{ N_b(\vecx)} \sum_{i=1}^n I[\mathbf{X}_i\in L_b(\vecx)] \cdot I(Y_i \leq y),
\end{equation}
where $I(\cdot)$ is the indicator function.
The WRF conditional CDF estimator averages the tree-level estimators:
\begin{equation}
\label{eq:ECDF}
\hat{F}(y\mid \vecx) = \frac{1}{B} \sum_{b=1}^B \hat{F}_b(y\mid \vecx) =\sum_{i=1}^n w_i(\vecx) \cdot I(Y_i \leq y),
\end{equation}
where
\begin{equation}
\label{eq:ECDF_weight}
w_i(\vecx)= \frac{1}{B} \sum_{b=1}^B \frac{ I[\mathbf{X}_i \in L_b(\vecx)]}{N_b(\vecx)}.
\end{equation}
These forest weights are nonnegative and satisfy $\sum_{i=1}^n w_i(\vecx)=1$. 
The WRF is different from the QRF by \cite{meinshausen2006} in tree building, in which we rely on Wilcoxon splitting criterion (\ref{eq:loss_rank}) instead of (\ref{eq:loss_SSE}) and we use subsampling instead of the bootstrap.  Both WRF and QRF construct empirical CDFs in individual trees and then average them over all the trees.  As with all random forests, WRF is an adaptive nearest-neighbor approach \citep{linjeon2006, scornetetal2015}, with its neighborhoods defined by rank-based Wilcoxon partitions.

With the estimated conditional CDF (\ref{eq:ECDF}), we can estimate conditional quantiles and conditional means.  The conditional $\tau$ quantile for $0<\tau<1$ can be obtained as
\begin{equation}
\label{eq:quantile}
\hat{Q}_\tau (\vecx)= \inf \{y: \hat{F}(y\mid \vecx) \geq \tau\}.
\end{equation}
Other approaches for obtaining quantiles from the estimated CDF can be found in \cite{hyndmanfan1996}, \cite{harrell1982new}, \cite{yang1985smooth}, \cite{sheather1990kernel}, and \cite{tian2024}. When all the outcomes $Y$ are fully observed, we can also estimate the conditional mean as
\begin{align*}
\hat{\mu}(\vecx)= \sum_{i=1}^n w_i(\vecx) Y_i.
\end{align*}

Since our node splitting criterion is rank-based with respect to the outcome, the WRF can be viewed as a rank-based nonparametric regression method. It is invariant to strictly increasing transformations of the outcome in the following sense: if $Z=h(Y)$ is a strictly increasing transformation of $Y$, then, under the same forest randomization and tuning parameters, the WRF fitted to $Z$ has exactly the same forest structure as the WRF fitted to $Y$, including the same splits, leaf memberships and thus, the same forest weights. Consequently, if $\hat{F}_Z( h(y) \mid \vecx)$ is the estimated CDF from the WRF of $Z$ on the covariates, then $\hat{F}_Z(h(y)\mid\ \vecx) =\hat{F}(y\mid \vecx)$ for any $y$.

Under the conditions stated in Supplementary Section~S1.3, for each fixed prediction point $\vecx$, the WRF distribution estimator in \eqref{eq:ECDF} satisfies
$$\sup_{y \in \real} \left| \hat{F} (y \mid \vecx)- F(y \mid \vecx) \right| \xrightarrow{p} 0.$$
The proof is provided in Supplementary Section~S1.3. This consistency result holds for any orderable outcome including continuous, count, ordinal, or mixture-type (e.g., continuous data subject to detection limits) outcomes.

While the CPM refinement in Section~\ref{sec:2_3-pred} can stabilize the noisy leaf-wise ECDFs of a single tree model, we do not use CPM-refined trees in the WRF. At the forest level, averaging ECDF-based trees already provides stabilization and yields the weighted empirical CDF estimator in \eqref{eq:ECDF}. Fitting a CPM on each tree would add additional assumptions and substantially increase computational cost. In exploratory simulations reported in Supplementary Section~S2.4, CPM refinement produced small, mixed changes in performance and did not provide a consistent overall improvement over the ECDF-based WRF.

\subsection{Evaluation Metrics and Variable Importance}
Because the WRF estimates a full conditional distribution, we evaluate its prediction performance at the distributional level.
The evaluation can be carried out on an independent test set, through cross-validation, or using out-of-bag (OOB) prediction.
For cross-validation, each observation is predicted by a model fitted using data outside its validation fold.
For OOB prediction of training observation $i$, let $\mathcal{B}_i^{\text{OOB}} = \{b: i \not\in \mathcal{I}_b\}$ denote the trees where training observation $i$ is out-of-bag. For each $b\in \mathcal{B}_i^{\text{OOB}}$, we compute the predicted CDF $\hat{F}_b(y \mid \mathbf{X}_i)$.
The OOB predicted distribution for observation $i$ is the average of these tree-level CDFs over $b \in \mathcal{B}_i^{\text{OOB}}$.
We use calibration diagnostics and proper scoring rules to summarize distributional predictive accuracy. In the context of distributional prediction, calibration refers to the coherence between the predicted distributions and the observed outcomes \citep{gneiting2007}.

For simplicity, we describe evaluation metrics in the context of an independent test set.
Let $\{(\mathbf{X}_k, Y_k)\}_{k=1}^m$ denote the observations in the test set, and let $\hat{F}_k(y) = \hat{F}(y \mid \mathbf{X}_k)$ denote their predicted CDFs. 
For fully observed continuous outcomes, distributional calibration can be assessed using the probability integral transform (PIT): $\mathrm{PIT}_k = \hat{F}_k(Y_k)$. Under a well-calibrated predictive distribution, PIT values should be approximately distributed as $\text{Uniform} (0,1)$ \citep{diebold1998evaluating, gneiting2007}. We can assess the overall calibration by quantifying the overall deviation of $\mathrm{PIT}_k$'s from $\text{Uniform} (0,1)$ using the Cram\'er-von Mises (CvM) statistic:
\begin{align*}
W^2= \sum_{k=1}^m \left( \mathrm{PIT}_{(k)}- \frac{2k-1}{2m} \right)^2 + \frac{1}{12m},
\end{align*}
where $\mathrm{PIT}_{(1)} \leq \cdots \leq \mathrm{PIT}_{(m)}$ are the ordered statistics of $\mathrm{PIT}_k$'s \citep{stephens1974edf}. A smaller value of the CvM statistic indicates a better calibration. 

Calibration at specific quantile levels can be summarized by empirical non-exceedance rates and central prediction intervals. Let $\hat{Q}_\tau (\mathbf{X}_k)$, ($\tau \in (0,1)$) be the estimated conditional $\tau$th quantile defined in (\ref{eq:quantile}). For $0<\tau_L<\tau_U<1$, define
\begin{align*}
\hat{p}_\tau = \frac{1}{m} \sum_{k=1}^m I[Y_k \leq \hat{Q}_\tau (\mathbf{X}_k)], \quad \widehat{\text{Coverage}}_{\tau_L, \tau_U} = \frac{1}{m} \sum_{k=1}^m I[\hat{Q}_{\tau_L} (\mathbf{X}_k) \leq Y_k \leq \hat{Q}_{\tau_U} (\mathbf{X}_k)].
\end{align*}
A well-calibrated model should have $\hat{p}_\tau$ close to $\tau$ and $\widehat{\text{Coverage}}_{\tau_L, \tau_U}$ close to $(\tau_U- \tau_L)$.

For a threshold $c$ of the outcome, we can evaluate the accuracy of the predicted probability $P(Y \leq c \mid \mathbf{X}= \vecx)$ using the Brier score \citep{brier1950}.
Specifically, 
\begin{align*}
\text{Brier}(c)= \frac{1}{m} \sum_{k=1}^{m} \left[ \hat{F}_k(c)- I(Y_k \leq c) \right]^2,
\end{align*}
where $\left[ \hat{F}_k(c)- I(Y_k \leq c) \right]^2$ is the threshold-specific per-observation Brier score loss.
To evaluate the full predictive distribution, we consider the continuous ranked probability score (CRPS; \citealp{crps1976}). For observation $k$, the CRPS is the threshold-specific Brier score loss integrated over all thresholds:
\begin{align*}
\text{CRPS} (\hat{F}_k, Y_k)= \int_{-\infty}^\infty \left[ \hat{F}_k(c)- I(Y_k \leq c) \right]^2 \mathrm{d}c. 
\end{align*}
Overall predictive performance can be summarized by averaging CRPS over the evaluation sample, with smaller values indicating better distributional prediction. 

We define permutation-based variable importance using OOB CRPS, so that a covariate's importance reflects its contribution to prediction of the full conditional distribution. For tree $T_b$, let $O_b= \{i: i \notin \mathcal{I}_b\}$ denote its OOB observations. Let $\hat F_{i|b}$ be the predicted CDF for $i \in O_b$. 
The baseline OOB loss for tree $T_b$ is
$$\overline{\mathrm{CRPS}}_b = \frac{1}{|O_b|} \sum_{i\in O_b} \mathrm{CRPS}(\hat F_{i|b},Y_i).$$
To evaluate the importance of covariate $X_j$, we randomly permute its values among observations in $O_b$, pass the permuted covariate vectors down the same tree $T_b$, and let $\hat F^{\pi}_{i|b,j}$ denote the resulting predicted CDF. The corresponding permuted OOB loss is
$$ \overline{\mathrm{CRPS}}^{\pi}_{b,j} = \frac{1}{|O_b|} \sum_{i\in O_b} \mathrm{CRPS}(\hat F^{\pi}_{i|b,j},Y_i).$$
If the permutation is repeated, $\overline{\mathrm{CRPS}}^{\pi}_{b,j}$ denotes the average over repetitions. We define the CRPS-based variable importance for covariate $X_j$ as the percentage increase in OOB CRPS:
\begin{equation}
\label{eq:vimp-crps-percent}
\widehat{\mathrm{VI}}^{\%}_j = 100\% \times \frac{ \sum_{b=1}^{B} \left( \overline{\mathrm{CRPS}}^{\pi}_{b,j} - \overline{\mathrm{CRPS}}_b \right) }{ \sum_{b=1}^{B} \overline{\mathrm{CRPS}}_b}.
\end{equation}
A positive $\widehat{\mathrm{VI}}^{\%}_j$ indicates that permuting $X_j$ worsens OOB distributional prediction, suggesting its relevance to the conditional distribution.

\section{Simulation Studies}
We conducted simulation studies to evaluate the performance of the Wilcoxon regression tree and WRF against existing methods that rely on SSE- or pinball-loss-based splitting rules.

\subsection{Simulation Setup}
We consider six independent random variables: $X_1 \sim N(0,1)$, $X_2 \sim \text{Bernoulli} (p=1/2)$, $X_3 \sim t(\text{df}=10)$, $X_4 \sim \text{Uniform} (0,1)$, $X_5 \sim \text{Beta} (\alpha=3, \beta=2)$, $X_6 \sim \text{Poisson} (\lambda=1)$. Here, $X_1, \ldots, X_4$ affect the outcome, and $X_5$ and $X_6$ act as noise.

\textbf{Scenario 1: Tree-Based Data Generating Process (DGP).}
Here, we used a tree-structured, piecewise constant DGP with seven terminal nodes (visualized in Supplementary Figure~S1). 
Let $c_{1,1}= \Phi^{-1} (1/6)$, $c_{1,2}= \Phi^{-1} (7/12)$, $c_3=F_{t, \text{df}=10}^{-1} (1/5)$ and $c_4=1/7$, where $\Phi(\cdot)$ is the standard normal CDF and $F_{t, \mathrm{df}=10}$ is the CDF of a Student's $t$ distribution with 10 degrees of freedom. A latent outcome $Y^*= \mu_1(\mathbf{X})+ \epsilon_1$, $\epsilon_1 \sim N(0,1)$ independent of $\mathbf{X}$, was generated with
$$\mu_1(\mathbf{X})= \begin{cases}
-1.5, & X_4 \leq c_4,\\
-1, & X_4 > c_4,\; X_1 \le c_{1,1},\\
-0.5, & X_4 > c_4,\; X_1 > c_{1,1},\; X_3 \le c_3 ,\\
0, & X_4 > c_4,\; X_1 > c_{1,1},\; X_3 > c_3,\; X_1 \le c_{1,2},\; X_2 = 1,\\
0.5, & X_4 > c_4,\; X_1 > c_{1,1},\; X_3 > c_3,\; X_1 \le c_{1,2},\; X_2 = 0,\\
1, & X_4 > c_4,\; X_1 > c_{1,1},\; X_3 > c_3,\; X_1 > c_{1,2},\; X_2 = 1,\\
1.5, & X_4 > c_4,\; X_1 > c_{1,1},\; X_3 > c_3,\; X_1 > c_{1,2},\; X_2 = 0. \end{cases}$$
The resulting partition has seven terminal regions, each with probability $1/7$. Two observed outcomes were derived from $Y^*$:
\begin{itemize}
\item \textit{Scenario 1.1 (Untransformed)}: $Y_{1.1}=Y^*$, and thus $Y_{1.1}\mid\mathbf{X} \sim N(\mu_1( \mathbf{X}), 1)$, with symmetric, homoscedastic errors. This setting was designed to favor SSE-based splitting. 

\item \textit{Scenario 1.2 (Transformed)}: $Y_{1.2}= h(Y^*)$, where $h(y)= F_{\chi^2_{15}}^{-1} [\Phi(y)]$, with $F_{\chi^2_{15}}^{-1}(\cdot)$ being the inverse CDF of the chi-squared distribution with 15 degrees of freedom. This transformation is  strictly increasing and produces a right-skewed marginal distribution with heteroscedastic, non-normal errors, while preserving the same outcome ranks as Scenario 1.1 \citep{tianetal2023}. Thus, the Wilcoxon regression tree and WRF should produce identical splits on  $Y_{1.1}$ and $Y_{1.2}$, while SSE-based and pinball-loss-based splitting may change after transformation. 
The true conditional mean was numerically approximated via $10^7$ Monte Carlo draws. 
\end{itemize}

\textbf{Scenario 2: Continuous Non-Linear DGP.}
To assess performance when the true regression function is smooth rather than tree-structured, we generated $Y_2=\mu_2(\mathbf{X})+ \sigma(\mathbf{X}) \cdot \epsilon_2$, where
$$\begin{cases} \mu_2(\mathbf{X})= 0.5+ X_1+ 0.8 X_2+ 0.3 X_3+ 0.7 \sin(2\pi X_4)- 0.5 X_1 X_2\\ \sigma(\mathbf{X})= 0.4+ 0.3 \cdot I(X_4>0.5)+ 0.5\cdot I(X_1 X_3>0) \end{cases},$$
and $\epsilon_2 = Z-3$, where $Z\sim \chi_3^2$ was independent of all covariates. Since $E(\epsilon_2)=0$ and $\Var(\epsilon_2)=6$, the conditional distributions are right-skewed and heteroscedastic with $\Var(Y_2\mid\mathbf{X})= 6\sigma(\mathbf{X})^2$. The true conditional distribution is $F_{Y_2\mid\mathbf{X}} (y)= F_{\chi_3^2} \Big[ \frac{y- \mu_2(\mathbf{X})}{ \sigma(\mathbf{X})} +3 \Big]$ and the conditional $\tau$th quantile is $Q_{Y_2\mid\mathbf{X}} (\tau\mid\mathbf{X})= \mu_2(\mathbf{X})+ \sigma(\mathbf{X}) \big[ F_{\chi_3^2}^{-1} (\tau)-3 \big]$.

For single-tree comparisons (Scenarios 1.1--1.2), we considered: (1) CART regression tree with SSE splitting; (2) Wilcoxon regression tree; (3) pinball-loss tree with the loss being the sum of pinball losses at $\tau= 0.1, 0.5, 0.9$; (4) CPM regression tree with probit link (correctly specified); and (5) CPM regression tree with logistic link (misspecified). For forest comparisons (all scenarios), we compared the WRF with the standard QRF. The pinball-loss forest was excluded because the pinball-loss trees were found to be substantially inferior to all other single-tree methods (Supplementary Tables~S1--S6) and required 4--8 times more computation time than the Wilcoxon tree.

For each simulation replicate, we drew training samples of size $n_{\text{train}} \in \{1000, 2000, 5000\}$. We used a fixed set of $n_{\text{test}}=1000$ covariate vectors to evaluate the results; given the covariate vectors, test outcomes were generated independently for each simulation replicate. We ran 1000 replicates for single-tree comparisons, and 250 for forest comparisons.
When generating single trees, we used all 6 predictors per split with minimum node size $n_{\text{node}}=60, 120, 300$, corresponding to $n_{\text{train}}=1000, 2000, 5000$, respectively. Forests were built on 1000 trees with 0.632 subsampling, $\textit{mtry}=2$, and $n_{\text{node}}=10, 20, 50$ respectively for Scenario 1, and $5, 10, 25$ respectively for Scenario 2.
For each test point and in all Scenarios, we estimated the conditional 0.1, 0.5, 0.9 quantiles, and the conditional mean. In Scenario 1, we also estimated the threshold probability $F(c \mid \mathbf{x})$ for $c=0$ (Scenario 1.1) or $c= F_{\chi_{15}^2}^{-1} [\Phi (0)] \approx 14.339$ (Scenario 1.2); these thresholds correspond to the same latent-scale value. We report RMSE and bias for these target quantities. We also report  mean CRPS and CvM statistic (as defined in Section \ref{sec:3-WRF}) for distributional prediction. Additional metrics (pinball losses, Brier scores, coverage rates, non-exceedance rates, median interval widths) are in the Supplementary~S2.2--S2.3. 
All simulations were implemented in R. Software versions and computing details for the reported execution times are provided in Supplementary Section~S2.1.

\subsection{Simulation Results}

\subsubsection{Single-Tree Comparisons}
\label{sec:4_3_1-tree}
Full single-tree results are reported in the Supplementary~S2.2. Here we summarize the key findings.
In Scenario 1.1, CART and the Wilcoxon regression tree had comparable performance across all metrics, with differences becoming negligible as $n_{\text{train}}$ increased. The CPM regression tree with correctly specified probit link achieved the best performance across all metrics and training sizes, because it combines the Wilcoxon tree's partitioning with semiparametric distributional smoothing across leaves. Even with the misspecified logistic link, the CPM tree performed well, although its calibration as measured by the CvM statistic was worse than the other methods. The pinball-loss tree was substantially worse than all other methods across all training sizes and almost all metrics, including the pinball loss it directly targeted.

In Scenario 1.2, the Wilcoxon regression tree improved over CART across all metrics: quantile RMSEs were reduced by 5--16\% depending on the quantile level and training size, and the threshold probability and mean RMSEs were also lower (Supplementary Tables~S4--S6). These improvements occurred because SSE splitting is sensitive to extreme values. The CPM regression tree with the probit link again achieved the best overall single-tree performance. In both scenarios, the Wilcoxon regression tree and CART took comparable computation time, while the pinball-loss tree was 4--8 times slower than the Wilcoxon tree; the CPM regression trees were much slower than all other trees.

\subsubsection{Forest Comparisons: WRF vs. QRF}
\label{sec:4_3_2-forest}
Tables~\ref{tab:forest_s12} and \ref{tab:forest_s2} display the results for the WRF and QRF under Scenarios 1.2 and 2. Full results, including biases, pinball losses, Brier scores, and calibration diagnostics are in Supplementary Section~S2.3. 

For Scenario 1.1 (untransformed outcomes), the WRF and QRF produced comparable results across all metrics and training sizes, with relative differences within 2\% (Supplementary Tables~S7--S9). These results showed that the WRF had little efficiency loss even though Scenario 1.1 was designed to favor the QRF. The WRF took approximately 35--60\% more computation time.

For Scenario 1.2 (transformed outcomes), the WRF outperformed the QRF across all metrics for all training sizes (Table~\ref{tab:forest_s12}). The WRF had reduced quantile RMSEs by 3.0--5.7\% at $\tau=0.1$, by 0.8--3.2\% at $\tau=0.5$, and by 7.9--15.7\% at $\tau=0.9$. The WRF also had reduced RMSEs for the conditional mean (3.7--8.4\%) and for the threshold probability (4.5--11.3\%). The CvM statistic was consistently lower for the WRF, indicating better calibration. The mean CRPS was also slightly lower for the WRF at each training size, with relative reductions of 0.2--0.3\%.

For Scenario 2 (continuous non-linear DGP), the WRF had consistent improvement in calibration: the CvM statistic was 9--10\% lower across all training sizes (Table~\ref{tab:forest_s2}). The WRF had lower quantile RMSEs at $\tau=0.1$ and $\tau=0.5$ (2--6\% reduction); the two methods had a similar quantile RMSE at $\tau=0.9$. The conditional mean RMSE was slightly lower for the WRF, and the mean CRPS was similar between the two methods.

\section{Application}
\label{sec:5-app}
\subsection{Clinical Motivation and Study Population}
We apply the WRF to predict the distribution of two biomarkers measured approximately six months after antiretroviral therapy (ART) initiation: CD4 cell count (cells/$\mu$L) and plasma HIV-1 RNA viral load (VL, copies/mL). CD4 cell count is a marker of immune status, with higher values indicating a healthier immune system.
VL quantifies ongoing viral replication, and the primary treatment goal is to achieve suppression below the assay detection limit \citep{NIH_VLMonitoring}. For CD4, we study the full conditional distribution and the threshold probabilities $P(\text{CD4} \le c \mid \mathbf{X})$ for $c =200, 350$ and $500$, corresponding to traditional thresholds denoting advanced, partial and minimal immunodeficiency, respectively. \citep{WHO_AHD} For VL, we focus on the probability of viral suppression, $P(\text{VL} < 80 \mid \mathbf{X})$.

The data come from a multicenter cohort study conducted in six Latin American countries from 2013 to 2023, enrolling $N=11,070$ ART-naive adults living with HIV. For each subject, we use the first CD4 and VL measurements within a $6 \pm 3$ month window after ART initiation. Ten covariates used in modeling are age at ART initiation, sex, site (Argentina, Brazil, Chile, Honduras, Mexico, Peru), mode of HIV infection (heterosexual or homosexual contact, other, unknown), prior AIDS-defining event, baseline CD4, baseline VL, ART regimen (INSTI-, NNRTI- or PI-based, other), months from ART initiation to measurement, and calendar year. Descriptive statistics for covariates are in the Supplementary Table~S16. 

The 6-month CD4 distribution is right-skewed. VL presents a different challenge: values below the assay-specific detection limit (DL) are reported only as ``$<$DL,'' with DLs varying across sites and years ($<$20, $<$40, $<$50, and $<$80 copies/mL). We harmonized to a common DL of 80 copies/mL and assigned all values below 80 a common tied rank; 83.1\% of subjects had VL below this threshold.

\subsection{6-Month CD4 Cell Count}
\label{sec:5_2}
We fit the WRF and, for comparison, the standard QRF using 0.632 subsampling, $n_{\text{tree}}=1000$ trees, and $\textit{mtry}=3$. The minimum node size was tuned over $\{5, 10, 20, 50, 100, 200\}$ using 10-fold cross-validation (CV): node size 5 gave the best performance for both methods across all metrics (Supplementary Section~S3.2). We evaluated predictive performance using the same 10-fold CV.

Table~\ref{tab:cd4_cv} reports the 10-fold CV performance. The WRF achieved lower Brier scores for $P(\text{CD4}\le 200 \mid \mathbf{X})$ and $P(\text{CD4}\le 350 \mid \mathbf{X})$, and a lower CvM statistic, indicating better calibration. 
The mean CRPS was slightly higher for the WRF than for the QRF, whereas the median CRPS was lower. This discrepancy suggested that the mean CRPS was influenced by a small number of observations with large CRPS values, while the WRF performed better for the typical subject as summarized by the median.
These results are consistent with the simulation findings in Section~\ref{sec:4_3_2-forest}. When the outcome distribution is skewed, the WRF improved probability estimation at clinically relevant thresholds and modestly improved calibration.
Additional diagnostic and illustrative results are provided in Supplementary Section~S3.3.

Table~\ref{tab:vimp_cd4} reports variable importance for predicting 6-month CD4, measured by the percentage increase in OOB mean CRPS defined in (\ref{eq:vimp-crps-percent}) after randomly permuting each covariate (averaged over 10 repetitions). Baseline CD4 was by far the most important predictor for both forests, followed by baseline VL and age. The rankings of variable importance were consistent for both methods.

Consider two reference profiles with all baseline covariates fixed at their cohort medians (continuous) or modes (categorical) except for baseline CD4 and VL, which are set to \{200 cells/$\mu$L, $10^6$ copies/mL\} and \{500 cells/$\mu$L, $10^4$ copies/mL\}, corresponding to starting ART with advanced versus less advanced baseline disease, respectively. Figure~\ref{fig:cd4vl_cdf}(A) displays the WRF-estimated conditional CDFs of 6-month CD4 for these two profiles. Let $\vecx_A$ and $\vecx_B$ denote the advanced and less advanced disease profiles, respectively. The distribution for the advanced disease profile was shifted towards lower 6-month CD4 counts, with estimated median 425 cells/$\mu$L and $\hat P(\text{CD4} \le 500 \mid \mathbf{X}= \vecx_A) = 0.67$, compared with 643 cells/$\mu$L for estimated median and $\hat P(\text{CD4} \le 500 \mid \mathbf{X}= \vecx_B) = 0.13$ for the less advanced profile.

\subsection{6-Month HIV Viral Load}
The vast majority, 83.1\%, of participants had 6-month VL below 80 copies/mL. Given the high left-censoring rate, the standard QRF cannot be applied. 
The WRF accommodates the detection limit by treating VL as a mixed continuous-ordinal outcome and assigning all values below the detection limit to a common lowest ranked level. 
Using the same forest settings as for CD4 ($n_{\text{tree}}=1000$, $\textit{mtry}=3$, 0.632 subsampling), we tuned the minimum node size over $\{20, 50, 100, 200, 500, 1000\}$: larger candidate values are used here because of the heavy ties. Node size 20 gave the best CV Brier score (Supplementary Section~S3.2).

We compared the WRF with two binary-outcome competitors for estimating $P(\text{VL}<80 \mid \mathbf{X})$: logistic regression (with restricted cubic splines for continuous covariates) and a classification random forest trained on $I(\text{VL}<80)$ using the same forest configuration and random seeds. The WRF achieved a 10-fold CV Brier score of 0.116, representing a 17.1\% improvement over an uninformative model ($\hat{p}(1-\hat{p}) = 0.140$) and an 9.3\% improvement over logistic regression (Brier $= 0.127$). The classification forest, trained directly on the binary endpoint, achieved a lower Brier score of 0.085. This is expected as the classification forest directly targets the binary estimand, whereas the WRF models the full ordered outcome. The WRF thus remains attractive when the analysis aims to estimate other quantities that require retaining the ordering among uncensored VL values in a unified framework (e.g., Figure~\ref{fig:cd4vl_cdf}), while the classification forest is preferable when interest is limited to a single threshold.

Table~\ref{tab:vimp_vl} reports variable importance for viral suppression, measured by the percentage increase in OOB Brier score for $P(\text{VL}<80 \mid \mathbf{X})$ after randomly permuting each covariate. Baseline VL was the most important predictor, followed by months to measurement, baseline CD4, and age. These findings are clinically plausible because subjects with higher pretreatment VL require longer to achieve suppression, and the probability of suppression increases with time on ART during the early treatment period.

Figure~\ref{fig:cd4vl_cdf}(B) shows the WRF-estimated conditional CDFs for the 6-month VL outcome for the same two reference profiles defined at the end of Section~\ref{sec:5_2}. Because measurements below the detection limit, 80 copies/mL, were represented by a common lowest ranked level, the jump at $<80$ corresponds to the estimated $\hat P(\mathrm{VL}<80\mid \mathbf{X})$. This probability was lower for the advanced-disease profile than for the less advanced profile (0.73 versus 0.93). The estimated conditional 0.95 quantile was 1130 copies/mL for the advanced-disease profile, compared with 157 copies/mL for the less advanced profile. Thus the WRF summarizes the probability of having undetectable VL and the distribution of detectable VL with a single ordered-outcome analysis.

\section{Discussion}
In this paper, we proposed a Wilcoxon regression tree and a Wilcoxon random forest (WRF) for conditional distributional prediction. The key idea is to use a rank-based Wilcoxon splitting criterion instead of the traditional, scale-dependent, mean-oriented SSE criterion. 
As a result, the proposed tree and forest are invariant to strictly increasing transformations of the outcome and can be applied to continuous, ordinal, or mixed outcomes, including outcomes subject to detection limits. 
The WRF combines rank-based partitioning with ECDF estimates within nodes across multiple trees to estimate the conditional distribution. The resulting WRF estimator is consistent for the conditional distribution, and it allows estimation of conditional quantiles, threshold probabilities, and means.

Our simulations show that the WRF performs comparably to the standard QRF when the outcome distribution is approximately symmetric and homoscedastic, and has improved quantile estimation and better calibration when outcomes are skewed or heteroscedastic. 
In the HIV application, the WRF improved over the QRF for prediction of 6-month CD4 threshold probabilities and distribution calibration. For viral load, the WRF accommodated the large mass at the detection limit and provided a unified ordered-outcome analysis, although a direct binary classification forest outperformed it when the estimand of interest was the probability of viral suppression.

Our work can be extended in the future. First, the WRF can be extended to more general forms of censored outcomes, including settings with multiple detection limits and right- or interval-censored survival times. 
Second, it would be desirable to develop inference procedures for WRF predictions, such as confidence intervals for conditional quantiles or threshold probabilities. An honest version of the WRF, which uses separate data for split selection and distribution estimation \citep{wagerathey2018}, may make such inference easier to justify and could also support causal extensions, such as estimating treatment-specific outcome distributions and heterogeneous distributional treatment effects. 

Overall, the WRF provides a simple and robust version of distributional random forests. It performs well when the outcome is skewed, heavy-tailed, ordinal, or affected by detection limits. We expect rank-based forest methods to be useful in a broad range of applications where scientifically meaningful prediction depends on the shape of the conditional outcome distribution, not just on the mean.

\section*{Data Availability Statement}
The data used in this manuscript are available upon request subject to cohort network approval. In accordance with the Caribbean, Central and South America network for HIV epidemiology (CCASAnet) principles of collaboration, applicants seeking the data need to submit concept sheets outlining the intended use of the data. All data requests will be reviewed by appropriate committee within the CCASAnet. After approval, transfer of data will likely require data use agreements. We have created a synthetic analysis cohort dataset with similar structure to the original data that can be readily downloaded with analysis codes at: \url{https://github.com/dshidanni/WilcoxonRandomForest}

\section*{Acknowledgments}
This work was supported by the National Institutes of Health (NIH) grant R01AI093234 (B.E.S. and C.L.) and the NIH-funded CCASAnet (U01AI069923; D.S. and B.E.S.). 
The authors thank the participants, clinical teams, investigators, data managers, and coordinating staff of CCASAnet, whose data made the HIV application possible. 
The content is solely the responsibility of the authors and does not necessarily represent the official views of the NIH, CCASAnet, or the participating institutions.

\bibliographystyle{chicago}
\bibliography{reference}

@article{tianetal2023,
	author = {Tian, Yuqi and Shepherd, Bryan E. and Li, Chun and Zeng, Donglin and Schildcrout, Jonathan S.},
	journal = {Biometrics},
	number = {79},
	pages = {3764-3777},
	title = {Analyzing clustered continuous response variables with ordinal regression models},
	year = {2023}}

@article{scornetetal2015,
	author = {Scornet, Erwan and Biau, Gerard and Vert, Jean-Philippe},
	journal = {The Annals of Statistics},
	number = {4},
	pages = {1716-1741},
	title = {Consistency of Random Forests},
	volume = {43},
	year = {2015}}

@article{meinshausen2006,
	author = {Meinshausen, Nicolai},
	journal = {Journal of Machine Learning Research},
	volume = {7},
	pages = {983-999},
	title = {Quantile Regression Forests},
	year = {2006}}

@article{hyndmanfan1996,
    author = {Hyndman, Rob J. and Fan, Yanan},
    title = {Sample Quantiles in Statistical Packages},
    journal = {The American Statistician},
    year = {1996},
    volume = {50},
    number = {4},
    pages = {361-365}}

@book{CARTbook,
    author = {Breiman, Leo and Friedman, Jerome and Olshen, R A. and Stone, Charles J.},
    title = {Classification and Regression Trees},
    publisher = {Chapman and Hall/CRC},
    year = 1984
}

@inproceedings{bhat2015,
  title={Towards Scalable Quantile Regression Trees},
  author={Bhat, Harish S and Kumar, Nitesh and Vaz, Garnet J},
  booktitle={2015 IEEE International Conference on Big Data (Big Data)},
  pages={53--60},
  year={2015},
  organization={IEEE},
  doi={10.1109/BigData.2015.7363741}
}

@article{koenker1978, title={Regression Quantiles},
  author={Koenker, Roger and Bassett, Gilbert},
  journal={Econometrica},
  volume={46},
  number={1},
  pages={33--50},
  year={1978},
  publisher={JSTOR},
  doi={10.2307/1913643}
}

@article{linjeon2006,
  title={Random Forests and Adaptive Nearest Neighbors},
  author={Lin, Yi and Jeon, Yongho},
  journal={Journal of the American Statistical Association},
  volume={101},
  number={474},
  pages={578--590},
  year={2006},
  publisher={Taylor \& Francis},
  doi={10.1198/016214505000001230}
}

@article{mccullagh1980,
  title={Regression Models for Ordinal Data},
  author={McCullagh, Peter},
  journal={Journal of the Royal Statistical Society, Series B},
  volume={42},
  number={2},
  pages={109--127},
  year={1980},
  publisher={JSTOR}
}

@article{liu2017,
  title={Modeling Continuous Response Variables Using Ordinal Regression},
  author={Liu, Qi and Shepherd, Bryan E. and Li, Chun and Harrell, Jr., Frank E.},
  journal={Statistics in Medicine},
  volume={36},
  number={27},
  pages={4316--4335},
  year={2017},
  publisher={Wiley Online Library},
  doi={10.1002/sim.7433}
}

@article{tian2024,
  title={Addressing Multiple Detection Limits with Semiparametric Cumulative Probability Models},
  author={Tian, Yuqi and Li, Chun and Tu, Shengxin and James, Nathan T. and Harrell, Frank E. and Shepherd, Bryan E.},
  journal={Journal of the American Statistical Association},
  volume={119},
  number={546},
  pages={864--874},
  year={2024},
  publisher={Taylor \& Francis},
  doi={10.1080/01621459.2024.2315667}
}

@article{breiman2001,
  title={Random Forests},
  author={Breiman, Leo},
  journal={Machine Learning},
  volume={45},
  number={1},
  pages={5--32},
  year={2001},
  publisher={Springer},
  doi={10.1023/A:1010933404324}
}

@article{harrell1982new,
  title={A New Distribution-Free Quantile Estimator},
  author={Harrell Jr, Frank E. and Davis, C. E.},
  journal={Biometrika},
  volume={69},
  number={3},
  pages={635--640},
  year={1982},
  publisher={Oxford University Press},
  doi={10.1093/biomet/69.3.635}
}

@article{yang1985smooth,
  title={A Smooth Nonparametric Estimator of a Quantile Function},
  author={Yang, Shie-Shien},
  journal={Journal of the American Statistical Association},
  volume={80},
  number={392},
  pages={1004--1011},
  year={1985},
  publisher={Taylor \& Francis},
  doi={10.1080/01621459.1985.10478217}
}

@article{sheather1990kernel,
  title={Kernel Quantile Estimators},
  author={Sheather, Simon J. and Marron, J. S.},
  journal={Journal of the American Statistical Association},
  volume={85},
  number={410},
  pages={410--416},
  year={1990},
  publisher={Taylor \& Francis},
  doi={10.1080/01621459.1990.10476214}
}

@manual{rms,
    title = {rms: Regression Modeling Strategies},
    author = {Frank E {Harrell Jr}},
    year = {2025},
    note = {R package version 8.0-0},
    url = {https://CRAN.R-project.org/package=rms},
    doi = {10.32614/CRAN.package.rms},
  }

@article{brier1950,
  author    = {Brier, Glenn W.},
  title     = {Verification of Forecasts Expressed in Terms of Probability},
  journal   = {Monthly Weather Review},
  volume    = {78},
  number    = {1},
  pages     = {1--3},
  year      = {1950},
  doi       = {10.1175/1520-0493(1950)078<0001:VOFEIT>2.0.CO;2}
}

@article{crps1976,
  author    = {Matheson, James E. and Winkler, Robert L.},
  title     = {Scoring Rules for Continuous Probability Distributions},
  journal   = {Management Science},
  volume    = {22},
  number    = {10},
  pages     = {1087--1096},
  year      = {1976},
  doi       = {10.1287/mnsc.22.10.1087}
}

@article{gneiting2007,
  author    = {Gneiting, Tilmann and Balabdaoui, Fadoua and Raftery, Adrian E.},
  title     = {Probabilistic Forecasts, Calibration and Sharpness},
  journal   = {Journal of the Royal Statistical Society. Series B, Statistical Methodology},
  volume    = {69},
  number    = {2},
  pages     = {243--268},
  year      = {2007},
  doi       = {10.1111/j.1467-9868.2007.00587.x}
}

@article{diebold1998evaluating,
  author    = {Diebold, Francis X. and Gunther, Todd A. and Tay, Anthony S.},
  title     = {Evaluating Density Forecasts with Applications to Financial Risk Management},
  journal   = {International Economic Review},
  volume    = {39},
  number    = {4},
  pages     = {863--883},
  year      = {1998},
  doi       = {10.2307/2527342}
}

@article{stephens1974edf,
  author    = {Stephens, M. A.},
  title     = {{EDF} Statistics for Goodness of Fit and Some Comparisons},
  journal   = {Journal of the American Statistical Association},
  volume    = {69},
  number    = {347},
  pages     = {730--737},
  year      = {1974},
  doi       = {10.1080/01621459.1974.10480196}
}

@Manual{r2026,
  title        = {R: A Language and Environment for Statistical Computing},
  author       = {{R Core Team}},
  organization = {R Foundation for Statistical Computing},
  address      = {Vienna, Austria},
  year         = {2026},
  url          = {https://www.R-project.org/}
}

@Manual{Matrix,
    title = {Matrix: Sparse and Dense Matrix Classes and Methods},
    author = {Douglas Bates and Martin Maechler and Mikael Jagan},
    year = {2025},
    note = {R package version 1.7-4},
    url = {https://CRAN.R-project.org/package=Matrix},
    doi = {10.32614/CRAN.package.Matrix},
  }

@misc{WHO_AHD,
  author       = {{World Health Organization}},
  title        = {Guidelines for Managing Advanced {HIV} Disease and Rapid Initiation of Antiretroviral Therapy},
  year         = {2017},
  month        = jul,
  url          = {https://www.who.int/publications/i/item/9789241550062},
  note         = {Guideline. Accessed 2026-04-21}
}

@misc{NIH_VLMonitoring,
  author       = {{Panel on Antiretroviral Guidelines for Adults and Adolescents}},
  title        = {Plasma {HIV}-1 {RNA} (Viral Load) and {CD4} Count Monitoring},
  year         = {2025},
  month        = sep,
  organization = {{ClinicalInfo, U.S. Department of Health and Human Services}},
  url          = {https://clinicalinfo.hiv.gov/en/guidelines/hiv-clinical-guidelines-adult-and-adolescent-arv/plasma-hiv-1-rna-cd4-monitoring},
  note         = {Updated September 25, 2025. Accessed 2026-04-21}
}

@article{chaudhuri2002nonparametric,
  title={Nonparametric Estimation of Conditional Quantiles Using Quantile Regression Trees},
  author={Chaudhuri, Probal and Loh, Wei-Yin},
  journal={Bernoulli},
  volume={8},
  number={5},
  pages={561--576},
  year={2002},
  publisher={Bernoulli Society for Mathematical Statistics and Probability}
}

@misc{wagerwalther2015,
  title={Adaptive Concentration of Regression Trees, with Application to Random Forests},
  author={Wager, Stefan and Walther, Guenther},
  year={2015},
  eprint={1503.06388},
  archivePrefix={arXiv},
  primaryClass={math.ST},
  doi={10.48550/arXiv.1503.06388}
}

@article{wagerathey2018,
  title={Estimation and Inference of Heterogeneous Treatment Effects using Random Forests},
  author={Wager, Stefan and Athey, Susan},
  journal={Journal of the American Statistical Association},
  volume={113},
  number={523},
  pages={1228--1242},
  year={2018},
  doi={10.1080/01621459.2017.1319839}
}

@article{stone1977,
  author  = {Stone, Charles J.},
  title   = {Consistent Nonparametric Regression},
  journal = {The Annals of Statistics},
  year    = {1977},
  volume  = {5},
  doi     = {10.1214/aos/1176343886}
}

@article{BuhlmannYu2002,
  author  = {Bühlmann, Peter and Yu, Bin},
  title   = {Analyzing bagging},
  journal = {The Annals of Statistics},
  year    = {2002},
  volume  = {30},
  doi     = {10.1214/aos/1031689014}
}

\clearpage
\section*{Tables}
\begin{table}[!ht]
\centering
\small
\caption{Forest comparisons under Scenario 1.2 ($Y_{1.2} = h(Y^*)$). Columns: RMSE for conditional quantile ($\tau = 0.1, 0.5, 0.9$), probability ($c \approx 14.339$), and mean estimates; mean CRPS; CvM statistic; computation time (s) for 1000 trees; averaged over 250 replicates with $n_{\text{test}} = 1000$. Bold indicates the best value in each column within each $n_{\text{train}}$.}
\label{tab:forest_s12}
\begin{tabular}{l ccc cc cc r}
\toprule
 & \multicolumn{3}{c}{Quantile RMSE} & Prob.\ & Mean & & & \\
\cmidrule(lr){2-4}
Method & $\tau{=}0.1$ & $\tau{=}0.5$ & $\tau{=}0.9$ & RMSE & RMSE & CRPS & CvM & Time (s) \\
\midrule
\multicolumn{9}{l}{\textit{$n_{\text{train}}=1000$, $n_{\text{node}}=10$}} \\[2pt]
QRF & 1.556 & 1.413 & 2.613 & 0.089 & 1.504 & 3.145 & 0.396 & 2.927 \\
WRF & \textbf{1.509} & \textbf{1.401} & \textbf{2.406} & \textbf{0.085} & \textbf{1.449} & \textbf{3.138} & \textbf{0.380} & 3.985 \\
\addlinespace[6pt]
\multicolumn{9}{l}{\textit{$n_{\text{train}}=2000$, $n_{\text{node}}=20$}} \\[2pt]
QRF & 1.339 & 1.111 & 2.178 & 0.073 & 1.247 & 3.107 & 0.363 & 3.427 \\
WRF & \textbf{1.282} & \textbf{1.091} & \textbf{1.940} & \textbf{0.068} & \textbf{1.178} & \textbf{3.100} & \textbf{0.343} & 5.025 \\
\addlinespace[6pt]
\multicolumn{9}{l}{\textit{$n_{\text{train}}=5000$, $n_{\text{node}}=50$}} \\[2pt]
QRF & 1.168 & 0.845 & 1.864 & 0.062 & 1.040 & 3.082 & 0.341 & 4.837 \\
WRF & \textbf{1.102} & \textbf{0.818} & \textbf{1.572} & \textbf{0.055} & \textbf{0.953} & \textbf{3.074} & \textbf{0.322} & 7.972 \\
\bottomrule
\end{tabular}
\end{table}

\begin{table}[!ht]
\centering
\small
\caption{Forest comparisons under Scenario 2 (continuous non-linear DGP, right-skewed heteroscedastic errors). Columns: RMSE for conditional quantile ($\tau = 0.1, 0.5, 0.9$) and mean estimates; mean CRPS; mean CvM; computation time (s); averaged over 250 replicates, $n_{\text{test}} = 1000$. Bold indicates the best value in each column within each $n_{\text{train}}$.}
\label{tab:forest_s2}
\begin{tabular}{l ccc c cc r}
\toprule
 & \multicolumn{3}{c}{Quantile RMSE} & Mean & & & \\
\cmidrule(lr){2-4}
Method & $\tau{=}0.1$ & $\tau{=}0.5$ & $\tau{=}0.9$ & RMSE & CRPS & CvM & Time (s) \\
\midrule
\multicolumn{8}{l}{\textit{$n_{\text{train}}=1000$, $n_{\text{node}}=5$}} \\[2pt]
QRF & 0.702 & 0.519 & 1.024 & 0.460 & 1.075 & 0.488 & 6.127 \\
WRF & \textbf{0.673} & \textbf{0.509} & \textbf{1.011} & \textbf{0.450} & \textbf{1.073} & \textbf{0.443} & 7.848 \\
\addlinespace[6pt]
\multicolumn{8}{l}{\textit{$n_{\text{train}}=2000$, $n_{\text{node}}=10$}} \\[2pt]
QRF & 0.649 & 0.453 & \textbf{0.859} & 0.386 & 1.062 & 0.467 & 6.668 \\
WRF & \textbf{0.617} & \textbf{0.437} & 0.860 & \textbf{0.375} & \textbf{1.060} & \textbf{0.419} & 8.625 \\
\addlinespace[6pt]
\multicolumn{8}{l}{\textit{$n_{\text{train}}=5000$, $n_{\text{node}}=25$}} \\[2pt]
QRF & 0.607 & 0.404 & \textbf{0.708} & 0.329 & 1.053 & 0.448 & 7.885 \\
WRF & \textbf{0.572} & \textbf{0.384} & 0.732 & \textbf{0.322} & \textbf{1.051} & \textbf{0.402} & 11.385 \\
\bottomrule
\end{tabular}
\end{table}

\begin{table}[!ht]
\centering
\small
\caption{10-fold CV performance for 6-month CD4 ($N = 11,070$, node size $= 5$, 1000 trees, 0.632 subsampling, $\text{mtry} = 3$). Bold = better value per column.}
\label{tab:cd4_cv}
\begin{tabular}{l cc c ccc}
\toprule
 & \multicolumn{2}{c}{CRPS} & & \multicolumn{3}{c}{Brier Score} \\
\cmidrule(lr){2-3} \cmidrule(lr){5-7}
Method & Mean & Median & CvM & $c=200$ & $c=350$ & $c=500$ \\
\midrule
QRF & \textbf{69.05} & 48.20 & 4.730 & 0.067 & 0.081 & 0.091 \\
WRF & 69.52 & \textbf{47.77} & \textbf{4.685} & \textbf{0.065} & \textbf{0.080} & 0.091 \\
\bottomrule
\end{tabular}
\end{table}

\begin{table}[!ht]
\centering
\small
\caption{Variable importance for 6-month CD4: percentage increase in OOB mean CRPS upon random permutation of each covariate, averaged over 10 repetitions.}
\label{tab:vimp_cd4}
\begin{tabular}{l rr}
\toprule
 & \multicolumn{2}{c}{\% Increase in CRPS} \\
\cmidrule(lr){2-3}
Covariate & QRF & WRF \\
\midrule
Baseline CD4       & 142.53 & 140.10 \\
Baseline VL        &  24.36 &  23.69 \\
Age                &  19.94 &  19.47 \\
Site               &  16.54 &  16.08 \\
Months to measurement   &  15.67 &  15.53 \\
Calendar year      &  12.21 &  11.77 \\
Prior AIDS event   &  11.24 &  11.07 \\
Infection mode     &   9.06 &   9.39 \\
ART regimen        &   6.67 &   6.45 \\
Sex                &   2.29 &   2.26 \\
\bottomrule
\end{tabular}
\end{table}

\begin{table}[!ht]
\centering
\small
\caption{Variable importance for viral suppression: percentage increase in OOB Brier score for $P(\text{VL} < 80 \mid \mathbf{X})$ upon random permutation of each covariate, averaged over 10 repetitions.}
\label{tab:vimp_vl}
\begin{tabular}{l r}
\toprule
Covariate & \% Increase in Brier score \\
\midrule
Baseline VL & 7.691 \\
Months to measurement\ & 4.560 \\
Baseline CD4 & 3.474 \\
Age & 3.308 \\
Site & 2.516 \\
Calendar year & 1.926 \\
ART regimen & 1.484 \\
Infection mode & 0.660 \\
Prior AIDS event & 0.291 \\
Sex & 0.067 \\
\bottomrule
\end{tabular}
\end{table}

\begin{figure}[!ht]
    \centering
    \includegraphics[width=1\linewidth]{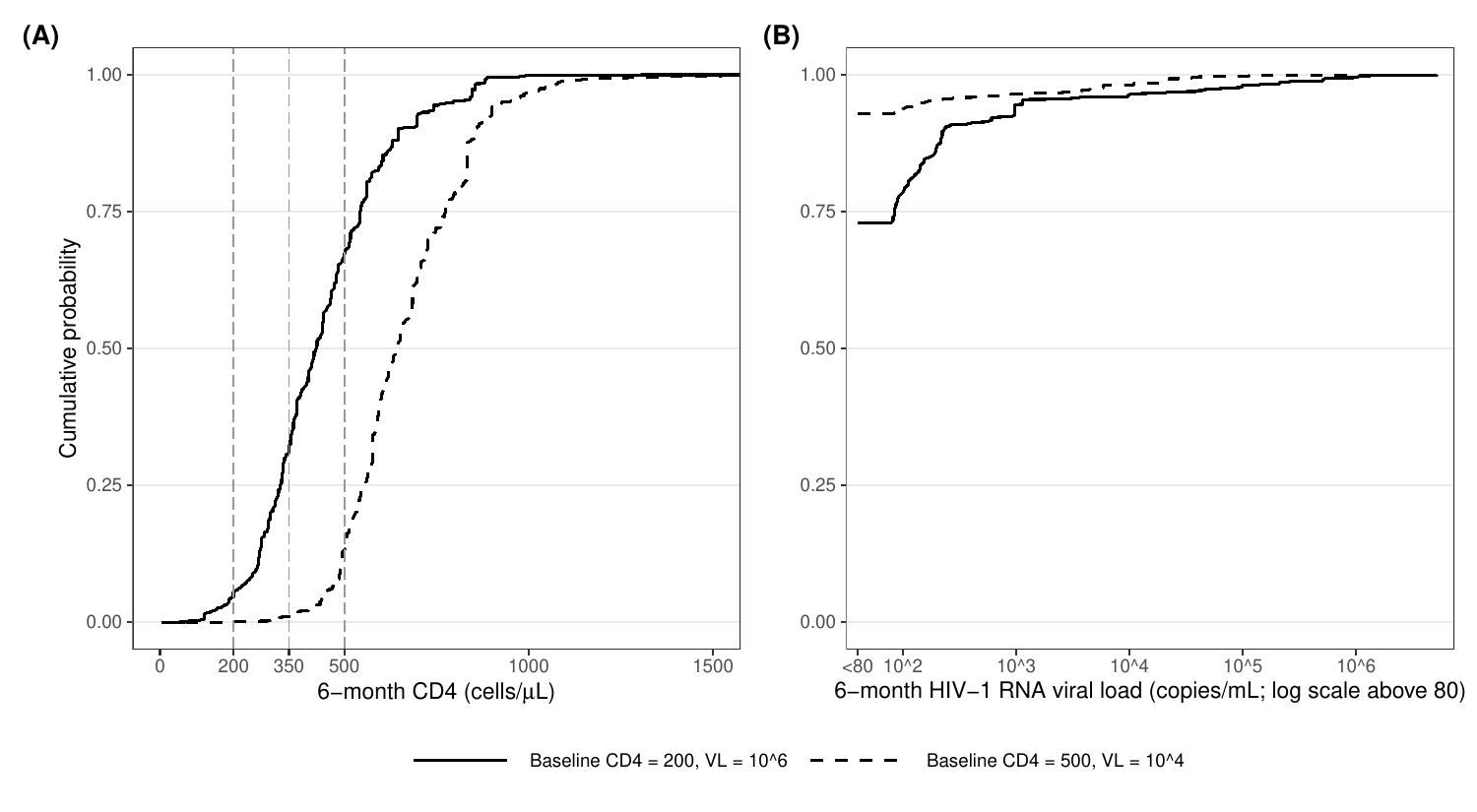}
    \caption{WRF-estimated conditional CDFs of (A) 6-month CD4 cell count and (B) 6-month HIV-1 RNA viral load for two hypothetical profiles that fix all covariates at their population medians (continuous) or modes (categorical) and differ only in baseline CD4 and VL. Forests were fit to the full cohort ($N=11,070$) with 1000 trees, $\textit{mtry}=3$, and minimum node sizes 5 for CD4 and 20 for VL.}
    \label{fig:cd4vl_cdf}
\end{figure}

\clearpage
\begin{center}
\textbf{\Large Supplementary Materials}
\end{center}
\setcounter{page}{1}
\setcounter{section}{0}
\setcounter{equation}{0}
\setcounter{table}{0}
\setcounter{figure}{0}
\renewcommand{\thepage}{S\arabic{page}}
\renewcommand{\thesection}{S\arabic{section}}
\renewcommand{\thetable}{S\arabic{table}}
\renewcommand{\thefigure}{S\arabic{figure}}
\renewcommand{\theequation}{S\arabic{equation}}

\section{Technical Details and Proofs}
\label{sec:supp_technical}

\subsection{Rank-Based Cost-Complexity Pruning}
\label{sec:supp_pruning}
A Wilcoxon regression tree structure can also be pruned using a rank-based cost-complexity measure analogous to that for the CART cost-complexity pruning \citep{CARTbook}. Let $T_0$ denote the large tree grown using the Wilcoxon splitting rule in (\ref{eq:loss_rank}). Let $T \subseteq T_0$ denote a subtree and $\mathcal{L}(T)$ denote the set of terminal nodes of $T$.
For pruning, we computed a fixed set of global midranks $a_i$ for the outcomes among all observations used to grow $T_0$. For any subtree $T$, define the global-rank impurity:
\begin{align*}
\mathcal{R}_{\text{rank}} (T)= \sum_{\ell \in \mathcal{L}(T)} \sum_{i: \vecx_i \in \ell} (a_i - \bar{a}_\ell)^2, \quad \bar{a}_\ell = \frac{1}{n_\ell} \sum_{i: \mathbf{X}_i \in \ell} a_i,
\end{align*}
where $\ell$'s are terminal nodes of subtree $T$ and $n_\ell$ is the number of observations in terminal node $\ell$.
The rank-based cost-complexity pruning criterion is
\begin{align*}
C_\lambda(T)= \mathcal{R}_{\text{rank}} (T) + \lambda |\mathcal{L}(T)|, \quad \lambda \geq 0,
\end{align*}
where $|\mathcal{L}(T)|$ is the number of terminal nodes in tree $T$, $\lambda$ denotes the cost assigned to each terminal node and $\lambda |\mathcal{L}(T)|$ is the complexity penalty. The pruned tree is selected from the nested sequence of subtrees minimizing $C_\lambda(T)$.
The tuning parameter $\lambda$ can be selected by cross-validation.

\subsection{Proof of Proposition \ref{thm:prop1}}
\label{pf:z}
\begin{proof}
Consider a fixed parent node $A$ with size $n_A$, and a candidate split $(j,s)$ inducing left and right child nodes $A_L = A_L(j,s)$ and $A_R = A_R(j,s)$ of sizes $n_L$ and $n_R$, where $n_L + n_R = n_A$. 
Let $\mathcal{I}_A= \{i: \mathbf{X}_i \in A\}$, $\mathcal{I}_L= \{i: \mathbf{X}_i \in A_L\}$ and $\mathcal{I}_R= \{i: \mathbf{X}_i \in A_R\}$.
Let $r_{i,A}$ denote the midrank of $Y_i$ among $\{Y_k: k\in \mathcal{I}_A\}$.
Let $\bar{r}_A$, $\bar{r}_L$, and $\bar{r}_R$ denote the mean midranks in $A$, $A_L$ and $A_R$, respectively. 
The rank-based impurity reduction for this split is
$$\Delta_{\text{rank}}= \sum_{i \in \mathcal{I}_A} (r_{i,A} - \bar{r}_A)^2- \left\{\sum_{i \in \mathcal{I}_L} (r_{i,A} - \bar{r}_L)^2 + \sum_{i \in \mathcal{I}_R} (r_{i,A} - \bar{r}_R)^2 \right\}.$$

Using ANOVA decomposition, we can have
$$\sum_{i \in \mathcal{I}_A} (r_{i,A} - \bar{r}_A)^2= \sum_{i \in \mathcal{I}_L} (r_{i,A} - \bar{r}_L)^2 + \sum_{i \in \mathcal{I}_R} (r_{i,A} - \bar{r}_R)^2+ n_L (\bar{r}_L- \bar{r}_A)^2 +n_R (\bar{r}_R- \bar{r}_A)^2$$
So, $$\Delta_{\text{rank}}= n_L (\bar{r}_L- \bar{r}_A)^2 +n_R (\bar{r}_R- \bar{r}_A)^2$$
Note that 
$$\bar{r}_A= \frac{n_L \bar{r}_L + n_R \bar{r}_R}{n_A} \Rightarrow \bar{r}_L- \bar{r}_A= \frac{n_R}{n_A} (\bar{r}_L- \bar{r}_R),\; \bar{r}_R- \bar{r}_A= -\frac{n_L}{n_A} (\bar{r}_L- \bar{r}_R).$$
Hence,
\begin{align*}
\Delta_{\text{rank}} &= n_L \left(\frac{n_R}{n_A}\right)^2 (\bar{r}_L - \bar{r}_R)^2
 +  n_R \left(\frac{n_L}{n_A}\right)^2 (\bar{r}_L - \bar{r}_R)^2 
= \frac{n_L n_R}{n_A} (\bar{r}_L - \bar{r}_R)^2.
\end{align*}

Let $S_L$ denote the rank sum of the left child node $A_L$, i.e., $S_L= \sum_{i \in \mathcal{I}_L} r_{i,A}$. Let $U_A$ denote the normalized Mann-Whitney U statistic for comparing outcomes in $A_L$ and $A_R$. $U_A$ is often calculated as $$U_A= \frac{S_L- \frac{1}{2} n_L (n_L+1)}{n_L n_R}.$$ 
Then by definition, $n_L \bar{r}_L=S_L$, and $n_R \bar{r}_R= n_A \bar{r}_A- S_L$. Then,
\begin{align*}
\bar{r}_L - \bar{r}_R
&= \frac{S_L}{n_L} - \frac{n_A \bar{r}_A - S_L}{n_R} 
= S_L\left(\frac{1}{n_L} + \frac{1}{n_R}\right) - \frac{n_A \bar{r}_A}{n_R} \\
&= \frac{n_A}{n_L n_R} S_L - \frac{n_A}{n_R} \bar{r}_A 
= \frac{n_A}{n_L n_R} \bigl(S_L - n_L \bar{r}_A\bigr).
\end{align*}
So, $$S_L- n_L \bar{r}_A= S_L - n_L \cdot \frac{n_A+1}{2}= n_L n_R \left(U_A -\frac{1}{2} \right).$$
Because using midranks guarantees that the total rank sum remains unchanged even in the presence of ties, we have $\bar{r}_A= \frac{1}{2} (n_A +1)$. Then, we can have 
\begin{align*}
\Delta_{\text{rank}} &= \frac{n_A}{ n_L n_R} (S_L- n_L \bar{r}_A)^2
= \frac{n_A}{n_L n_R} \cdot n_L^2 n_R^2 \left(U_A- \frac{1}{2} \right)^2 
= n_A n_L n_R \left( U_A- \frac{1}{2} \right)^2
\end{align*}

Under the null hypothesis that the left and right node distributions are identical and observations are independent, $E(U_A)= \frac{1}{2}$ and the variance of $U_A$ is
$$\Var(U_A)= \frac{n_A+1- \frac{\kappa}{ n_A (n_A-1)}}{ 12 n_L n_R},$$
where $\kappa= \sum_{k=1}^K (t_k^3- t_k)$ is the tie-correction factor for tied rank groups of sizes $t_1, \cdots, t_K$. So, when both child-node sizes are sufficiently large, under the null hypothesis, the standardized Wilcoxon rank sum test statistic is $$z= \frac{U_A- \frac{1}{2}}{ \sqrt{\Var (U_A)}} \xrightarrow{d} N(0,1).$$
Consequently, we will have
$$\Delta_{\text{rank}}= \frac{n_A}{12} \left[ n_A+1 -\frac{\kappa}{ n_A (n_A-1)} \right] z^2$$

Therefore, within a fixed parent node, the rank-based impurity reduction $\Delta_{\text{rank}}$ is proportional to $z^2$, and maximizing $\Delta_{\text{rank}}$ is equivalent to maximizing $z^2$ and minimizing the associated Wilcoxon p-value approximated by $p_W(z) \approx 2\{1-\Phi(|z|)\}$, where $\Phi(\cdot)$ is the CDF of the standard normal distribution. The $p$-value is used for splitting and not used for formal inference.
\end{proof}

\subsection{Consistency of the WRF Distribution Estimator}
Let $\mathcal{D}_n= \{( \mathbf{X}_i, Y_i): 1\leq i \leq n\}$ be an i.i.d. sample from the distribution of $(\mathbf{X}, Y)$. Suppose $p$ (the dimension of covariate space) and $B$ (number of trees in the forest) are fixed as $n\to \infty$. For simplicity, the consistency result in this section is stated for numerical covariates rescaled to $\mathbf{X}_i \in \mathcal{X} = [0, 1]^p$. Let $F(y \mid \vecu) := F_Y(y \mid \mathbf{X}= \vecu)$.

Following the notations in Section~\ref{sec:3_1-WRFalgo} in the main text, for tree $T_b$, let $\mathcal{I}_b \subset \{1, \ldots, n\}$ denote the subsample used to construct $T_b$. $\mathcal{I}_b$ are drawn independently of $\mathcal{D}_n$ and are used to select the splits and and define the tree partition. Once the tree has been constructed, all $n$ training observations are passed through the fitted tree to estimate the terminal-node CDFs.
Let $L_b(\vecx)$ denote the terminal leaf node containing $\vecx$, and $N_b(\vecx)= \sum_{i=1}^n I\{ \mathbf{X}_i \in L_b(\vecx)\}$.
The tree-level and WRF conditional CDF estimators, $\hat{F}_b(y \mid \vecx)$ and $\hat{F}(y \mid \vecx)$, are defined in Equations \eqref{eq:tree_ECDF} and \eqref{eq:ECDF}, respectively.
The nonnegative forest weights $w_i(\vecx)$'s are defined in Equation \eqref{eq:ECDF_weight} and satisfy $\sum_{i=1}^n w_i(\vecx)=1$.

The following assumptions are motivated by the consistency argument for quantile regression forests \citep{meinshausen2006} and the conditions in Theorem~1 of \cite{stone1977}.

\begin{assumption}[Leaf size and vanishing maximal weight]
\label{ass:leafsize_weights}
Each terminal leaf corresponds to a rectangular region obtained from recursive binary splits on the covariates. Let $\mathcal{L}(T_b)$ denote the set of terminal leaf nodes of tree $T_b$. For $\ell\in \mathcal{L}(T_b)$, define $N_{b\ell}=\sum_{i=1}^n I(\mathbf{X}_i\in \ell)$, and let
$$ N_{\min,n} = \min_{1\leq b\leq B}\min_{\ell\in \mathcal{L}(T_b) } N_{b\ell}, 
\qquad N_{\max,n} = \max_{1\leq b\leq B}\max_{\ell \in \mathcal{L}(T_b)} N_{b\ell}.$$
There is a deterministic sequence $\alpha_n$ such that $$\frac{\alpha_n}{ \log n} \to \infty, \quad P(N_{\min, n} \geq \alpha_n) \xrightarrow{p} 1, \quad \frac{N_{\max, n}}{n} \xrightarrow{p} 0.$$ 
So, the minimum size of a leaf node increases with the sample size, and the maximum leaf-node size as a proportion of the total sample size vanishes.

Let $\omega_n(\vecx):= \max_{1\leq i\leq n} w_i(\vecx)$ denote the maximal WRF weight. When $N_{\min, n} \geq \alpha_n$, $$0< \omega_n(\vecx) \leq N_{\min, n}^{-1} \leq \alpha_n^{-1}.$$
Since $P(N_{\min,n} \geq \alpha_n) \to 1$ and $\alpha_n \to \infty$, we have $\omega_n (\vecx) \xrightarrow{p}0$. 
\end{assumption}


\begin{assumption}[Local continuity in Kolmogorov distance]
\label{ass:localcont}
For a fixed prediction point $\vecx$, define
$$\omega(\vecx, r)= \sup_{\vecu \in \mathrm{supp} (P_\mathbf{X}), \Vert \vecu- \vecx \Vert_1 \leq r} \sup_{y \in \real} |F(y \mid \vecu)- F(y \mid \vecx)|,$$
for some $r>0$. Then, $\omega(\vecx, r) \to 0$ as $r \to 0$.
\end{assumption}

\begin{assumption}[Pointwise localization of the forest weights]
\label{ass:nearx}
For every fixed $\vecx\in \mathrm{supp}(P_\mathbf{X}) \subseteq \mathcal X$ and $\forall \epsilon>0$,
$$ R_n(\epsilon,\vecx) := \sum_{i=1}^n w_i(\vecx) I\{ \lVert\mathbf{X}_i-\vecx \rVert_1>\epsilon \} \xrightarrow{p} 0.$$
\end{assumption}

Assumption \ref{ass:leafsize_weights} guarantees that every terminal leaf contains enough observations, preventing any single observation from dominating the weighted empirical distribution. 
This condition is sufficient for the uniform concentration of the empirical CDF across all possible splits.
Assumption \ref{ass:localcont} is a local continuity smoothness condition in Komogorov distance on the conditional distribution for observed outcomes. Assumption \ref{ass:nearx} requires the total weight assigned to observations with covariates separated from $\vecx$ to vanish. This condition is motivated by the localization conditions for probability-weighted local averages established by \cite{stone1977}.
We impose Assumption~\ref{ass:nearx} directly on the WRF weights, and we do not claim that it follows from the Wilcoxon splitting criterion alone.

For the proof of WRF consistency, we separate the estimation error into a localization component and a component for empirical CDFs. Assumptions~\ref{ass:localcont} and \ref{ass:nearx} control the localization term by ensuring that most of the forest weight is placed on covariate values close to the fixed $\vecx$, where the conditional CDFs are close to $F(y \mid \vecx)$. 

The component for empirical CDFs requires an argument of adaptive concentration because the outcomes in each tree-building subsample are used to select the splits, and these observations also contribute to the empirical CDF within the resulting terminal node. Thus, the terminal-leaf memberships depend on the outcomes, and a concentration bound for a prespecified leaf cannot directly apply. 
We address this issue in Lemma~\ref{lem:lemma} by establishing a uniform concentration bound over all possible axis-aligned rectangles. We do not track all possible split sequences, and instead, conditional on the observed covariates and subsample index sets, we consider all index sets induced on each tree subsample by axis-aligned rectangular regions. 

Given the observed covariates, for tree $T_b$ and an axis-aligned rectangular region $A \subseteq \mathcal{X}$, let $J_b(A)= \{i \in \{1, \ldots, n\}: \mathbf{X}_i \in A \}$ be the set of indices of all training observations that fall in $A$, where $A$ is some axis-aligned rectangular region under $\mathcal{X}$. Let $\mathcal{J}_{b,n}= \{J_b(A): A \text{ is a axis-aligned rectangular region}\}$ denote the collection of all distinct index sets of such form. Because recursive axis-aligned splitting always yields rectangular terminal nodes, $\mathcal{J}_{b,n}$ covers the full-sample membership of every terminal leaf that tree $T_b$ could select, regardless of how the splits are determined.

For any covariate, the $n$ observed values divide the real line into at most $n +1$ intervals. Moving an endpoint within one of these intervals does not change which training observations are selected, and the membership changes only when the endpoint crosses an observed value. So, each of the lower and upper endpoints for the intervals takes at most $n+1$ relevant positions, and an interval along one covariate can induce at most $(n +1)^2$ distinct subsets of observations. For $p$ covariates, we can have $|\mathcal{J}_{b,n}| \leq (n+1)^{2p}$, which is an upper bound on the number of distinct sample memberships. 
This argument is inspired by the adaptive concentration approach in \cite{wagerwalther2015}, which simultaneously controls estimation error over all leaves that may be selected using the outcomes. Our approach differs as we consider empirical CDFs rather than with-leaf outcome means. Thus, Lemma~\ref{lem:lemma} uniformly controls the empirical CDF error over all leaf memberships in $\mathcal{J}_{b, n}$'s that contain at least $\alpha_n$ observations.

\begin{lemma}[Adaptive concentration of empirical CDFs over possible leaf memberships]
\label{lem:lemma}
Under Assumption~\ref{ass:leafsize_weights}, define
$$\gamma_n := \max_{1 \leq b \leq B} \max_{J \in \mathcal{J}_{b,n}: |J| \geq \alpha_n} \sup_{y \in \real} \left| \frac{1}{|J|} \sum_{i \in J} [I(Y_i \leq y)- F(y \mid \mathbf{X}_i)] \right|.$$
Then, $\gamma_n \xrightarrow{p} 0$.
\end{lemma}
\begin{proof}
Now, fix $\epsilon>0$. Since $\frac{\alpha_n }{\log n} \to \infty$, we have $\alpha_n \to \infty$. So, for all sufficiently large $n$, $\alpha_n \geq \max \{2, 2/\epsilon\}$. We hereafter consider such $n$ for this proof. 

Let $\mathcal{H}_n= \sigma (\mathbf{X}_1, \ldots, \mathbf{X}_n, \mathcal{I}_1, \ldots, \mathcal{I}_B)$ denote information contained in all the observed covariates and the subsample index sets for all trees in the forest. Conditional on $\mathcal{H}_n$, the outcomes $Y_1, \ldots, Y_n$ are independent, with $P(Y_i \leq y \mid \mathcal{H}_n)= E[I(Y_i \leq y) \mid \mathcal{H}_n]= F(y \mid \mathbf{X}_i)$.

Fix a nonempty $J$ considered in the definition of $\gamma_n$, let $m= |J|$ such that $m \geq \alpha_n$, and thus, $m\geq 2$ and $m \epsilon \geq 2$. Let $F_i(y)= F(y \mid \mathbf{X}_i)$, and let $$\gamma_J= \sup_{y \in \real} \left| \frac{1}{m} \sum_{i \in J} [I(Y_i \leq y) - F_i(y)] \right|.$$ 
Let $\hat{F}_J(y)= \frac{1}{m} \sum_{i\in J} I(Y_i \leq y)$. Let $\bar{F}_J(y)= \frac{1}{m} \sum_{i \in J} F_i(y)$. So $\gamma_J= \sup_{y \in \real} |\hat{F}_J(y) - \bar{F}_J(y)|$.

Note that $\bar{F}_J(y)$ is non-decreasing, and $\hat{F}_J(y)$ is a non-decreasing step function, where jumps occur only at observed $Y_i$, $i \in J$.
Additionally, define $$F_i(t^-)= \lim_{s \uparrow t} F_i(s), \quad \hat{F}_J(t^-)= \frac{1}{m} \sum_{i \in J} I(Y_i < t), \quad \bar{F}_J(t^-)= \frac{1}{m} \sum_{i \in J} F_i (t^-).$$
Let $Y_{(i)}$ denote the $i$th order statistic. In the interval $[Y_{(i)}, Y_{(i+1)})$ ($i \in J$), $\hat{F}_J$ is flat and only takes value at $Y_{(i)}$, and $\bar{F}_J$ is non-decreasing. Thus, in addition to $Y_{(i)}$, we also need to consider $Y_{(i+1)}^-$, the left limit of $Y_{(i+1)}$. So,
\begin{align*}
P(\gamma_J > \epsilon \mid \mathcal{H}_n) &\leq \underbrace{ \sum_{k \in J} P(|\hat{F}_J(Y_k) - \bar{F}_J(Y_k)| > \epsilon \mid \mathcal{H}_n) }_{ (\#)} + \underbrace{ \sum_{k \in J} P(|\hat{F}_J(Y_k^-) - \bar{F}_J(Y_k^-)| > \epsilon \mid \mathcal{H}_n) }_{ (\#\#)}.
\end{align*}

For (\#), we have
\begin{align*}
m[\hat{F}_J(Y_k) - \bar{F}_J(Y_k)] &= I(Y_k= Y_k) + \sum_{i\in J, i\neq k} I(Y_i \leq Y_k) - F_k(Y_k)- \sum_{i\in J, i\neq k} F_i(Y_k)\\
&= 1-F_k (Y_k)+ \sum_{i\in J, i\neq k} [I(Y_i \leq Y_k) - F_i(Y_k)].
\end{align*}
Note that $$E[I(Y_i \leq Y_k) \mid \mathcal{H}_n, Y_k]= P(Y_i \leq Y_k \mid \mathcal{H}_n, Y_k) = F_i(Y_k).$$
Let $D_{Jk}= \sum_{i\in J, i\neq k} [I(Y_i \leq Y_k) - F_i(Y_k)]$. Then, $E(D_{Jk})=0$. So,
\begin{align*}
P(|\hat{F}_J(Y_k) &- \bar{F}_J(Y_k)| >\epsilon \mid \mathcal{H}_n, Y_k) = P[|1- F_k(Y_k)+ D_{Jk}| >m \epsilon \mid \mathcal{H}_n, Y_k]\\
\leq& P(D_{Jk}> m\epsilon-1 + F_k(Y_k) \mid \mathcal{H}_n, Y_k) 
+ P(D_{Jk}< -m\epsilon-1 + F_k(Y_k) \mid \mathcal{H}_n, Y_k).
\end{align*}
For $i\in J, i\neq k$, when $Y_i \leq Y_k$, $I(Y_i \leq Y_k)- F_i(Y_k)= 1-F_i(Y_k)$; when $Y_i > Y_k$, $I(Y_i \leq Y_k)- F_i(Y_k)= -F_i(Y_k)$. Then, $1-F_i(Y_k) + F_i(Y_k)=1$. So, by Hoeffding's inequality,
\begin{align*}
P(D_{Jk}> m\epsilon-1 + F_k(Y_k) \mid \mathcal{H}_n, Y_k) &\leq \exp \left[ -\frac{ 2(m\epsilon-1 +F_k(Y_k))^2 }{m-1} \right].
\end{align*}
Similarly,
\begin{align*}
P(D_{Jk}< -m\epsilon-1 + F_k(Y_k) \mid \mathcal{H}_n, Y_k) &\leq \exp \left[ -\frac{ 2(m\epsilon+1 -F_k(Y_k))^2 }{m-1} \right].
\end{align*}
Since $0\leq F_k(Y_k) \leq 1$, a mutual upper bound for the two probabilities above is $\exp \left[ -\frac{ 2(m\epsilon -1)^2}{ m-1} \right]$. So,
$$P(|\hat{F}_J(Y_k)- \bar{F}_J(Y_k)| >\epsilon \mid \mathcal{H}_n, Y_k) \leq 2\exp \left[ -\frac{ 2(m\epsilon -1)^2}{ m-1} \right].$$
Then,
\begin{align*}
P(|\hat{F}_J(Y_k) - \bar{F}_J(Y_k)| >\epsilon \mid \mathcal{H}_n) &= E \left\{ P\left[ |\hat{F}_J(Y_k) - \bar{F}_J(Y_k)| >\epsilon \mid \mathcal{H}_n, Y_k \right] \Big| \mathcal{H}_n \right\}\\
&\leq 2\exp \left[ -\frac{ 2(m\epsilon -1)^2}{ m-1} \right].
\end{align*}
So
$$(\#) =\sum_{k\in J} P(|\hat{F}_J(Y_k) - \bar{F}_J(Y_k)| >\epsilon \mid \mathcal{H}_n) \leq 2m\exp \left[ -\frac{ 2(m\epsilon -1)^2}{ m-1} \right].$$

Then, for (\#\#), we have $$m[\hat{F}_J(Y_k^-) - \bar{F}_J(Y_k^-)] = -F_k(Y_k^-) + \sum_{i\in J, i\neq k} [I(Y_i <Y_k)- F_i(Y_k^-)].$$
Let $D_{Jk}^-= \sum_{i\in J, i\neq k} [I(Y_i <Y_k)- F_i(Y_k^-)]$. Note that, $$E[I(Y_i < Y_k) \mid \mathcal{H}_n, Y_k]= P(Y_i <Y_k \mid \mathcal{H}_n, Y_k) = F_i(Y_k^-).$$
Similar to the previous part, by Hoeffding's inequality, we have
\begin{align*}
P(|\hat{F}_J(Y_k^-) &- \bar{F}_J(Y_k^-)| >\epsilon \mid \mathcal{H}_n, Y_k) = P(|-F_k(Y_k^-) + D_{Jk}^-| > m\epsilon \mid \mathcal{H}_n, Y_k)\\
&\leq P(D_{Jk}^- > m\epsilon+ F_k(Y_k^-) \mid \mathcal{H}_n, Y_k^-)+ P(D_{Jk}^- < -m\epsilon+ F_k(Y_k^-) \mid \mathcal{H}_n, Y_k)\\
&\leq \exp \left[ -\frac{2 (m\epsilon + F_k(Y_k^-))^2 }{m-1} \right] + \exp \left[ -\frac{2 (m\epsilon - F_k(Y_k^-))^2 }{m-1} \right]\\
&\leq 2\exp \left[ -\frac{2 (m\epsilon -1)^2 }{m-1} \right].\\
\Rightarrow \sum_{k\in J} &P \left( |\hat{F}_J(Y_k^-)- \bar{F}_J(Y_k^-)| > \epsilon \mid \mathcal{H}_n \right) \leq 2m \exp \left[ -\frac{2 (m\epsilon -1)^2 }{m-1} \right].
\end{align*}
Thus, combining the results for (\#) and (\#\#), we have $$P(\gamma_J > \epsilon \mid \mathcal{H}_n) \leq (\#) + (\#\#) \leq 4m \exp \left[ -\frac{2 (m\epsilon -1)^2 }{m-1} \right].$$
Note that 
\begin{align*}
\frac{2 (m\epsilon -1)^2 }{m-1} &= 2\epsilon^2 (m-1)+ 4\left( \epsilon- \frac{1}{2} \right)^2 -1 + \frac{2 (\epsilon-1)^2}{ m-1} \\
&\geq 2\epsilon^2 (m-1)-1 \geq 2\epsilon^2 (\alpha_n -1) -1.
\end{align*}

Recall that $|\mathcal{J}_{b,n}| \leq (n+1)^{2p}$ and $m \leq n$. Then, since $\frac{\alpha_n}{ \log n} \to \infty$,
\begin{align*}
P(\gamma_n >\epsilon \mid \mathcal{H}_n) &\leq \sum_{b=1}^B \sum_{J\in \mathcal{J}_{b,n}, |J| \geq \alpha_n} 4|J| \exp \left[ -\frac{2 (|J| \epsilon -1)^2 }{ |J|-1} \right] \\
&\leq 4Bn(n+1)^{2p} \exp \left[1- 2\epsilon^2 (\alpha_n -1) \right]\\
&= 4Bn(n+1)^{2p} \cdot n^{\frac{1 +2 \epsilon^2 }{\log n} -2\epsilon^2 \frac{\alpha_n}{ \log n}} \to 0.
\end{align*}
Hence, $\gamma_n \xrightarrow{p} 0$. 
\end{proof}

Using the concentration bound in Lemma~\ref{lem:lemma} and Assumptions~\ref{ass:leafsize_weights}--\ref{ass:nearx}, we could derive the following consistency result for the WRF conditional CDF estimator.

\begin{theorem}[Consistency of the WRF distribution estimator]
\label{thm:consistency}
Suppose Assumptions~\ref{ass:leafsize_weights}--\ref{ass:nearx} hold. Then, for every fixed $\vecx \in \mathcal{X}$,
$$ \sup_{y \in \real} \left| \hat F(y \mid \mathbf{X}= \vecx)-F_Y(y \mid \mathbf{X}=\vecx) \right| \xrightarrow{p}0.$$
\end{theorem}
\begin{proof}
For notational convenience, let $F_i(y)= F(y \mid \mathbf{X}_i)$, let $F_\vecx(y)= F(y \mid \mathbf{X}= \vecx)$ and let $\hat{F}_\vecx (y)= \hat{F}(y \mid \mathbf{X}= \vecx)$. Then,
\begin{align*}
\sup_{y \in \real} |\hat{F}(y \mid \vecx)&- F_\vecx (y)| = \sup_{y\in \real} \left| \sum_{i=1}^n w_i(\vecx) [I(Y_i \leq y)- F_\vecx(y)] \right|\\
&\leq \sup_{y\in \real} \bigg| \underbrace{ \sum_{i=1}^n w_i(\vecx) [F_i(y)- F_\vecx (y)] }_{=a_n(y)} \bigg| + \sup_{y\in \real} \bigg| \underbrace{ \sum_{i=1}^n w_i(\vecx) [I(Y_i \leq y) -F_i(y)] }_{=b_n(y)} \bigg|.
\end{align*}

For $a_n(y)$, we have,
\begin{align*}
\sup_{y\in \real} |a_n(y)| &\leq \sum_{i=1}^n w_i(\vecx) \cdot \underbrace{ \sup_{y \in \real} |F_i(y)- F_\vecx(y)| }_{(*)}.
\end{align*}
Let $\omega(\vecx, r)$ be defined as in Assumption~\ref{ass:localcont}. When $\Vert \mathbf{X}_i -\vecx \Vert_1 \leq r$, $(*) \leq \omega( \vecx, r) \cdot I(\Vert \mathbf{X}_i -\vecx \Vert_1 \leq r)$.
When $\Vert \mathbf{X}_i -\vecx \Vert_1 >r$, $(*) \leq I(\Vert \mathbf{X}_i -\vecx \Vert_1 >r)$. So,
\begin{align*}
\sup_{y\in \real} |a_n(y)| &\leq \omega(\vecx, r) \sum_{i=1}^n w_i(\vecx) \cdot I(\Vert \mathbf{X}_i- \vecx \Vert_1 \leq r)+ \sum_{i=1}^n w_i(\vecx) \cdot I(\Vert \mathbf{X}_i- \vecx \Vert_1 >r)\\
&\leq \omega(\vecx, r)+ R_n(r, \vecx).
\end{align*}
For all $\epsilon>0$, choose $r$ such that $\omega(\vecx, r) < \frac{\epsilon}{2}$. Then, by Assumption~\ref{ass:nearx},
\begin{align*}
P\left( \sup_{y \in \real} |a_n(y)| >\epsilon \right) &\leq P[\omega (\vecx, r)+ R_n(\vecx, r)> \epsilon]\\
&\leq P\left( R_n(\vecx, r) > \frac{\epsilon}{2} \right) \to 0 \Rightarrow \sup_{y \in \real} |a_n(y)| \xrightarrow{p} 0.
\end{align*}

Now, for $b_n(y)$, by Equation~\eqref{eq:ECDF_weight} of the main text,
\begin{align*}
\sup_{y\in \real} |b_n(y)| &\leq \frac{1}{B} \sum_{b=1}^B \sup_{y\in \real} \left| \frac{1}{N_b(\vecx)} \sum_{i=1}^n [I(Y_i \leq y)- F_i(y)] \right|\\
&\leq \frac{1}{B} \sum_{b=1}^B \gamma_n = \gamma_n,
\end{align*}
where $\gamma_n$ is defined in Lemma~\ref{lem:lemma}. Note that $\gamma_n$ only controls for sets $J$'s satisfying $|J| \geq \alpha_n$, where $\alpha_n$ is the deterministic sequence defined in Assumption~\ref{ass:leafsize_weights}. So, we need to consider the cases when $N_{\min, n} \geq \alpha_n$ and when $N_{\min, n} < \alpha_n$. Then, $\forall \epsilon>0$,
\begin{align*}
P\left( \sup_{y\in \real} |b_n(y)| > \epsilon \right) =& \underbrace{ P\left( \left\{ \sup_{y\in \real} |b_n(y)| > \epsilon \right\} \cap \{N_{\min, n} \geq \alpha_n\} \right) }_{ \leq P(\gamma_n > \epsilon) \to 0 } \\
&+ \underbrace{ P\left( \left\{ \sup_{y\in \real} |b_n(y)| > \epsilon \right\} \cap \{N_{\min, n} < \alpha_n\} \right) }_{ \leq P(N_{\min, n} < \alpha_n) \to 0 }\\
\to& \;0 \Rightarrow \sup_{y \in \real} |b_n(y)| \xrightarrow{p} 0.
\end{align*}
Therefore, $ \sup_{y \in \real} \left| \hat F(y \mid \mathbf{X}= \vecx)-F_Y(y \mid \mathbf{X}=\vecx) \right| \xrightarrow{p}0.$
\end{proof}

The consistency results applies to continuous outcomes, discrete ordinal outcomes, and continuous outcomes subject to detection limits. Here we show that for ordinal discrete outcomes, the estimated conditional probability mass is also consistent. For any fixed point $a \in \mathcal{Y}$, let $\hat{P} (a \mid \vecx) = \hat{F}(a \mid \vecx)- \hat{F}(a^- \mid \vecx)$ denote the conditional probability mass at point $a$, where $F(a^- \mid \vecx)= \lim_{t \uparrow a} F(t \mid \vecx)$. Let $\hat{P}(a \mid \vecx)= \hat{F} (a \mid \vecx)- \hat{F}(a^- \mid \vecx)$ denote the estimated probability mass at $a$. Then,
\begin{align*}
|\hat{P} (a\mid \vecx)- P(a\mid \vecx)| &\leq |\hat{F}(a \mid \vecx) - F(a \mid \vecx)| + |\hat{F}(a^- \mid \vecx) - F(a^- \mid \vecx)|\\
&\leq 2\sup_{y \in \real} |\hat{F}(y \mid \vecx) - F(y \mid \vecx)|.
\end{align*}
Therefore, $\hat{P} (a\mid \vecx) \xrightarrow{p} P(a\mid \vecx)$.

The consistency result also applies to continuous outcomes subject to detection limits.
For instance, let $Y^*$ be a continuous latent outcome and suppose $Y$ is the observed outcome. Suppose the observed $Y$ is subject to a lower detection limit $d_L$ and an upper detection limit $d_U$ ($d_U> d_L$) such that $Y= \min\{ \max\{Y^*, d_L\}, d_U \}$. Then, 
\begin{align*}
&P(Y= d_L \mid \mathbf{X}= \vecx)= P(Y^* \leq d_L \mid \mathbf{X}= \vecx),\\
&F_Y(t \mid \mathbf{X}= \vecx) = F_{Y^*}(t \mid \mathbf{X}= \vecx), \quad d_L \leq t< d_U,\\
&P(Y=d_U \mid \mathbf{X}= \vecx)= P(Y^* \geq d_U \mid \mathbf{X}= \vecx).
\end{align*}

\section{Additional Simulation Details and Results}
\label{sec:supp_simulation}

\subsection{Data-Generating Processes and Simulation Design}
\label{sec:supp_dgp}

Figure~\ref{fig:DGP} displays the tree-structured data-generating process for the conditional mean function $\mu_1(\mathbf{X})$ in Scenario~1. The partition yields seven terminal nodes, each with probability $1/7$, using splits on $X_4$, $X_1$, $X_3$, and $X_2$. The noise covariates $X_5$ and $X_6$ do not appear in the tree. The splitting constants are $c_4=1/7$, $c_{1,1}=\Phi^{-1}(1/6)$, $c_{1,2}=\Phi^{-1}(7/12)$, and $c_3=F_{t_{10}}^{-1}(1/5)$.

\begin{figure}[!htbp]
    \centering
    \IfFileExists{DGS_tree_cont.pdf}{%
    \includegraphics[width=0.75\linewidth]{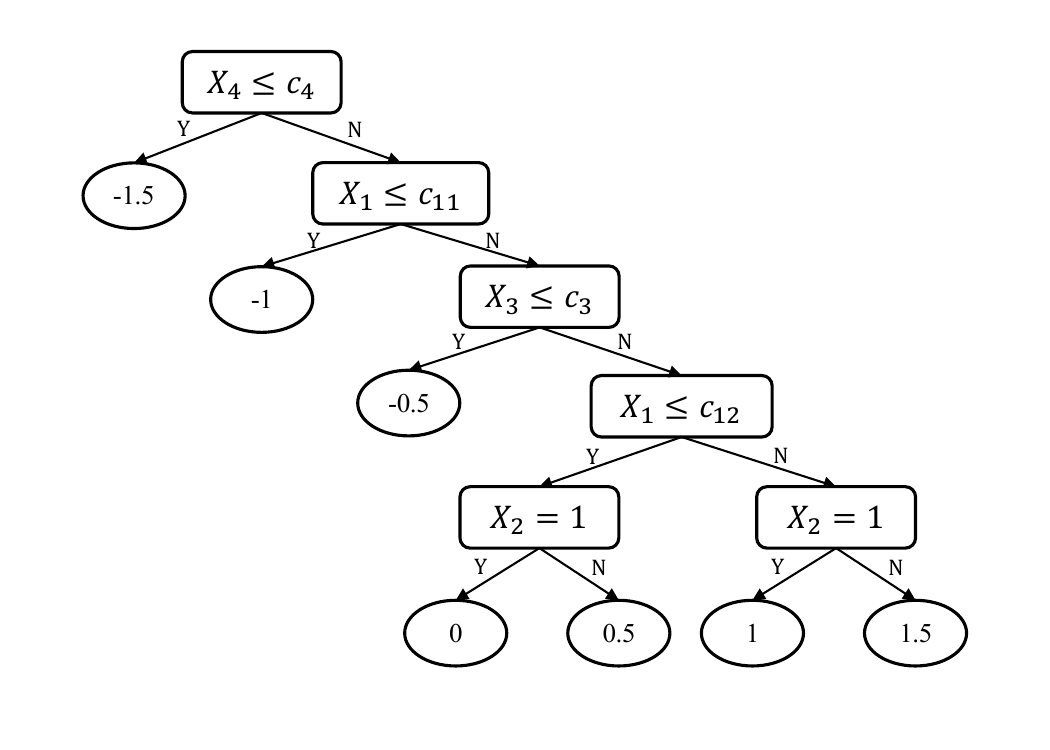}%
    }{%
    \fbox{\parbox{0.75\linewidth}{\centering Placeholder for \texttt{DGS\_tree\_cont.pdf}.}}%
    }
    \caption{Tree-structured data-generating process for $\mu_1(\mathbf{X})$ in Scenario~1. Each terminal node is labeled with its conditional mean value. The partition yields seven subspaces, each with probability $1/7$.}
    \label{fig:DGP}
\end{figure}

All simulation results below are averaged over independent Monte Carlo replicates with a fixed test-set size of $n_{\mathrm{test}}=1000$. Single-tree comparisons use 1000 replicates; forest comparisons use 250 replicates and 1000 trees per forest. For Scenarios~1.1 and~1.2, probability estimation targets $P(Y\leq c\mid \mathbf{X})$ with $c=0$ for Scenario~1.1 and $c=F^{-1}_{\chi^2_{15}}\{\Phi(0)\}\approx 14.339$ for Scenario~1.2. Bold values in the tables indicate the best value within each training size; for bias, coverage, and non-exceedance rates, ``best'' means closest to the target value.

All code was implemented in \textsf{R} 4.5.3 \citep{r2026} with use of the \texttt{rms} package for CPM-based approaches (version 8.1-1, \cite{rms}), and the \texttt{Matrix} package for matrix computations (version 1.7-4, \cite{Matrix}). Timing was performed on an Intel Core Ultra 7-265 processor with 32 GB RAM without parallelization.

\subsection{Full Single-Tree Results}
\label{sec:supp_tree}

Tables~\ref{tab:supp_tree_s11_quantile}--\ref{tab:supp_tree_s11_distribution} report full single-tree results for Scenario~1.1, and Tables~\ref{tab:supp_tree_s12_quantile}--\ref{tab:supp_tree_s12_distribution} report the corresponding results for Scenario~1.2. The compared methods are CART regression tree with SSE splitting (SSE), Wilcoxon regression tree, pinball-loss tree targeting $\tau\in\{0.1,0.5,0.9\}$, CPM tree with probit link, and CPM tree with logistic link. To avoid overly compressed tables, point-estimation and distributional metrics are shown in separate tables.

\begin{table}[!htbp]
\centering
\begingroup
\footnotesize
\setlength{\tabcolsep}{4pt}
\singlespacing
\caption{Full single-tree results under Scenario~1.1: conditional quantile RMSE, bias, and pinball loss.}
\label{tab:supp_tree_s11_quantile}
\begin{tabular}{lrrrrrrrrr}
\toprule
 & \multicolumn{3}{c}{$\tau=0.1$} & \multicolumn{3}{c}{$\tau=0.5$} & \multicolumn{3}{c}{$\tau=0.9$} \\
\cmidrule(lr){2-4}\cmidrule(lr){5-7}\cmidrule(lr){8-10}
Method & RMSE & Bias & Pinball & RMSE & Bias & Pinball & RMSE & Bias & Pinball \\
\midrule
\multicolumn{10}{l}{\textit{$n_{\mathrm{train}}=1000$, $n_{\mathrm{node}}=60$}} \\
SSE & 0.287 & 0.021 & 0.183 & 0.253 & -0.009 & 0.411 & 0.287 & -0.031 & 0.183 \\
Wilcoxon & 0.286 & 0.022 & 0.183 & 0.264 & -0.009 & 0.412 & 0.287 & -0.033 & 0.183 \\
Pinball & 0.368 & \textbf{-0.005} & 0.187 & 0.334 & -0.004 & 0.419 & 0.369 & \textbf{0.006} & 0.187 \\
CPM--probit & \textbf{0.250} & 0.026 & \textbf{0.181} & \textbf{0.244} & 0.003 & \textbf{0.410} & \textbf{0.250} & -0.018 & \textbf{0.181} \\
CPM--logistic & 0.256 & 0.046 & \textbf{0.181} & 0.251 & \textbf{0.002} & 0.411 & 0.254 & -0.037 & 0.182 \\
\addlinespace[3pt]
\multicolumn{10}{l}{\textit{$n_{\mathrm{train}}=2000$, $n_{\mathrm{node}}=120$}} \\
SSE & 0.202 & \textbf{0.009} & 0.179 & 0.179 & -0.004 & \textbf{0.405} & 0.202 & -0.013 & 0.179 \\
Wilcoxon & 0.200 & 0.010 & 0.179 & 0.186 & -0.004 & 0.406 & 0.199 & -0.015 & 0.179 \\
Pinball & 0.314 & -0.015 & 0.184 & 0.293 & \textbf{0.001} & 0.414 & 0.317 & 0.021 & 0.184 \\
CPM--probit & \textbf{0.174} & 0.014 & \textbf{0.178} & \textbf{0.170} & \textbf{0.001} & \textbf{0.405} & \textbf{0.173} & \textbf{-0.010} & \textbf{0.178} \\
CPM--logistic & 0.180 & 0.033 & \textbf{0.178} & 0.175 & \textbf{0.001} & \textbf{0.405} & 0.178 & -0.026 & 0.179 \\
\addlinespace[3pt]
\multicolumn{10}{l}{\textit{$n_{\mathrm{train}}=5000$, $n_{\mathrm{node}}=300$}} \\
SSE & 0.127 & \textbf{0.003} & \textbf{0.177} & 0.112 & -0.002 & 0.402 & 0.126 & -0.007 & \textbf{0.177} \\
Wilcoxon & 0.124 & \textbf{0.003} & \textbf{0.177} & 0.116 & -0.002 & 0.402 & 0.123 & -0.008 & \textbf{0.177} \\
Pinball & 0.295 & -0.024 & 0.183 & 0.283 & 0.003 & 0.413 & 0.299 & 0.036 & 0.183 \\
CPM--probit & \textbf{0.108} & 0.006 & \textbf{0.177} & \textbf{0.105} & \textbf{0.000} & \textbf{0.401} & \textbf{0.108} & \textbf{-0.005} & \textbf{0.177} \\
CPM--logistic & 0.114 & 0.022 & \textbf{0.177} & 0.109 & -0.001 & 0.402 & 0.113 & -0.020 & \textbf{0.177} \\
\bottomrule
\end{tabular}
\endgroup
\end{table}

\begin{table}[!htbp]
\centering
\begingroup
\footnotesize
\setlength{\tabcolsep}{5pt}
\singlespacing
\caption{Full single-tree results under Scenario~1.1: conditional mean and threshold-probability accuracy. The threshold is $c=0$.}
\label{tab:supp_tree_s11_meanprob}
\begin{tabular}{lrrrrr}
\toprule
 & \multicolumn{2}{c}{Mean} & \multicolumn{3}{c}{Threshold probability} \\
\cmidrule(lr){2-3}\cmidrule(lr){4-6}
Method & RMSE & Bias & RMSE & Bias & Brier \\
\midrule
\multicolumn{6}{l}{\textit{$n_{\mathrm{train}}=1000$, $n_{\mathrm{node}}=60$}} \\
SSE & \textbf{0.238} & \textbf{0.000} & 0.078 & -0.001 & 0.157 \\
Wilcoxon & 0.242 & \textbf{0.000} & 0.080 & \textbf{0.000} & 0.157 \\
Pinball & 0.315 & 0.005 & 0.101 & -0.002 & 0.161 \\
CPM--probit & 0.242 & 0.003 & \textbf{0.073} & -0.001 & \textbf{0.156} \\
CPM--logistic & 0.245 & 0.003 & 0.075 & -0.001 & \textbf{0.156} \\
\addlinespace[3pt]
\multicolumn{6}{l}{\textit{$n_{\mathrm{train}}=2000$, $n_{\mathrm{node}}=120$}} \\
SSE & \textbf{0.167} & \textbf{0.000} & 0.054 & \textbf{0.000} & \textbf{0.153} \\
Wilcoxon & 0.168 & \textbf{0.000} & 0.055 & \textbf{0.000} & 0.154 \\
Pinball & 0.282 & 0.006 & 0.089 & -0.002 & 0.158 \\
CPM--probit & 0.168 & 0.002 & \textbf{0.051} & \textbf{0.000} & \textbf{0.153} \\
CPM--logistic & 0.171 & 0.003 & 0.052 & \textbf{0.000} & \textbf{0.153} \\
\addlinespace[3pt]
\multicolumn{6}{l}{\textit{$n_{\mathrm{train}}=5000$, $n_{\mathrm{node}}=300$}} \\
SSE & \textbf{0.104} & -0.001 & 0.034 & \textbf{0.000} & \textbf{0.152} \\
Wilcoxon & \textbf{0.104} & -0.001 & 0.034 & \textbf{0.000} & \textbf{0.152} \\
Pinball & 0.278 & 0.007 & 0.086 & -0.003 & 0.158 \\
CPM--probit & \textbf{0.104} & \textbf{0.000} & \textbf{0.031} & \textbf{0.000} & \textbf{0.152} \\
CPM--logistic & 0.107 & 0.001 & 0.033 & \textbf{0.000} & \textbf{0.152} \\
\bottomrule
\end{tabular}
\endgroup
\end{table}

\begin{table}[!htbp]
\centering
\begingroup
\footnotesize
\setlength{\tabcolsep}{4pt}
\singlespacing
\caption{Full single-tree results under Scenario~1.1: distributional prediction metrics and computation time.}
\label{tab:supp_tree_s11_distribution}
\begin{tabular}{lrrrrrrrr}
\toprule
 & \multicolumn{2}{c}{CRPS} & & \multicolumn{2}{c}{80\% PI} & \multicolumn{2}{c}{Non-exceedance} & \\
\cmidrule(lr){2-3}\cmidrule(lr){5-6}\cmidrule(lr){7-8}
Method & Mean & Median & CvM & Coverage & Med.\ width & $\hat p_{0.1}$ & $\hat p_{0.9}$ & Time (s) \\
\midrule
\multicolumn{9}{l}{\textit{$n_{\mathrm{train}}=1000$, $n_{\mathrm{node}}=60$}} \\
SSE & 0.584 & 0.422 & 0.435 & 0.774 & 2.504 & 0.112 & 0.886 & \textbf{0.004} \\
Wilcoxon & 0.585 & 0.422 & 0.446 & 0.773 & 2.501 & 0.112 & 0.886 & 0.005 \\
Pinball & 0.677 & 0.493 & 0.377 & 0.777 & 2.558 & \textbf{0.111} & 0.888 & 0.022 \\
CPM--probit & \textbf{0.579} & \textbf{0.416} & \textbf{0.361} & \textbf{0.780} & 2.518 & \textbf{0.111} & \textbf{0.891} & 0.111 \\
CPM--logistic & 0.580 & 0.417 & 0.534 & 0.772 & \textbf{2.489} & 0.115 & 0.887 & 0.108 \\
\addlinespace[3pt]
\multicolumn{9}{l}{\textit{$n_{\mathrm{train}}=2000$, $n_{\mathrm{node}}=120$}} \\
SSE & 0.574 & 0.415 & 0.274 & 0.787 & 2.538 & \textbf{0.106} & 0.893 & \textbf{0.004} \\
Wilcoxon & 0.574 & 0.415 & 0.278 & 0.787 & 2.536 & \textbf{0.106} & 0.893 & 0.008 \\
Pinball & 0.671 & 0.489 & 0.274 & 0.787 & 2.591 & 0.107 & 0.894 & 0.034 \\
CPM--probit & \textbf{0.572} & \textbf{0.412} & \textbf{0.241} & \textbf{0.790} & 2.538 & \textbf{0.106} & \textbf{0.895} & 0.667 \\
CPM--logistic & 0.573 & 0.413 & 0.359 & 0.783 & \textbf{2.515} & 0.109 & 0.892 & 0.753 \\
\addlinespace[3pt]
\multicolumn{9}{l}{\textit{$n_{\mathrm{train}}=5000$, $n_{\mathrm{node}}=300$}} \\
SSE & 0.568 & 0.412 & 0.205 & 0.795 & 2.553 & \textbf{0.102} & 0.897 & \textbf{0.008} \\
Wilcoxon & 0.568 & 0.412 & 0.206 & 0.795 & 2.553 & \textbf{0.102} & 0.897 & 0.016 \\
Pinball & 0.665 & 0.484 & 0.222 & 0.795 & 2.614 & 0.104 & \textbf{0.899} & 0.063 \\
CPM--probit & \textbf{0.567} & \textbf{0.411} & \textbf{0.196} & \textbf{0.796} & 2.552 & \textbf{0.102} & 0.898 & 0.164 \\
CPM--logistic & 0.568 & 0.412 & 0.279 & 0.790 & \textbf{2.531} & 0.105 & 0.895 & 0.189 \\
\bottomrule
\end{tabular}
\endgroup
\end{table}

\begin{table}[!htbp]
\centering
\begingroup
\footnotesize
\setlength{\tabcolsep}{4pt}
\singlespacing
\caption{Full single-tree results under Scenario~1.2: conditional quantile RMSE, bias, and pinball loss.}
\label{tab:supp_tree_s12_quantile}
\begin{tabular}{lrrrrrrrrr}
\toprule
 & \multicolumn{3}{c}{$\tau=0.1$} & \multicolumn{3}{c}{$\tau=0.5$} & \multicolumn{3}{c}{$\tau=0.9$} \\
\cmidrule(lr){2-4}\cmidrule(lr){5-7}\cmidrule(lr){8-10}
Method & RMSE & Bias & Pinball & RMSE & Bias & Pinball & RMSE & Bias & Pinball \\
\midrule
\multicolumn{10}{l}{\textit{$n_{\mathrm{train}}=1000$, $n_{\mathrm{node}}=60$}} \\
SSE & 1.283 & 0.052 & 0.834 & 1.635 & -0.036 & 2.227 & 2.477 & -0.095 & 1.170 \\
Wilcoxon & 1.141 & 0.116 & 0.829 & 1.453 & -0.022 & 2.208 & 2.072 & -0.193 & 1.146 \\
Pinball & 1.484 & \textbf{-0.050} & 0.847 & 1.934 & \textbf{-0.021} & 2.253 & 2.872 & \textbf{0.079} & 1.186 \\
CPM--probit & \textbf{0.995} & 0.120 & \textbf{0.823} & \textbf{1.350} & 0.036 & \textbf{2.198} & \textbf{1.830} & -0.098 & \textbf{1.134} \\
CPM--logistic & 1.027 & 0.178 & 0.824 & 1.388 & 0.044 & 2.201 & 1.843 & -0.251 & 1.136 \\
\addlinespace[3pt]
\multicolumn{10}{l}{\textit{$n_{\mathrm{train}}=2000$, $n_{\mathrm{node}}=120$}} \\
SSE & 0.882 & \textbf{0.022} & 0.817 & 1.138 & -0.017 & 2.185 & 1.722 & \textbf{-0.035} & 1.133 \\
Wilcoxon & 0.788 & 0.052 & 0.815 & 1.024 & \textbf{-0.006} & 2.175 & 1.444 & -0.085 & 1.120 \\
Pinball & 1.267 & -0.084 & 0.834 & 1.676 & -0.013 & 2.228 & 2.422 & 0.151 & 1.161 \\
CPM--probit & \textbf{0.684} & 0.063 & \textbf{0.811} & \textbf{0.938} & 0.017 & \textbf{2.169} & \textbf{1.266} & -0.057 & \textbf{1.113} \\
CPM--logistic & 0.710 & 0.113 & 0.812 & 0.968 & 0.027 & 2.171 & 1.294 & -0.191 & 1.114 \\
\addlinespace[3pt]
\multicolumn{10}{l}{\textit{$n_{\mathrm{train}}=5000$, $n_{\mathrm{node}}=300$}} \\
SSE & 0.517 & \textbf{0.011} & 0.807 & 0.668 & -0.007 & 2.155 & 1.022 & \textbf{-0.028} & 1.107 \\
Wilcoxon & 0.485 & 0.020 & 0.806 & 0.638 & -0.006 & 2.154 & 0.892 & -0.044 & 1.103 \\
Pinball & 1.198 & -0.127 & 0.830 & 1.607 & -0.008 & 2.217 & 2.224 & 0.222 & 1.148 \\
CPM--probit & \textbf{0.420} & 0.025 & \textbf{0.805} & \textbf{0.579} & \textbf{0.004} & \textbf{2.151} & \textbf{0.786} & -0.030 & \textbf{1.101} \\
CPM--logistic & 0.441 & 0.069 & 0.806 & 0.601 & 0.011 & 2.152 & 0.831 & -0.157 & \textbf{1.101} \\
\bottomrule
\end{tabular}
\endgroup
\end{table}

\begin{table}[!htbp]
\centering
\begingroup
\footnotesize
\setlength{\tabcolsep}{5pt}
\singlespacing
\caption{Full single-tree results under Scenario~1.2: conditional mean and threshold-probability accuracy. The threshold is $c\approx 14.339$.}
\label{tab:supp_tree_s12_meanprob}
\begin{tabular}{lrrrrr}
\toprule
 & \multicolumn{2}{c}{Mean} & \multicolumn{3}{c}{Threshold probability} \\
\cmidrule(lr){2-3}\cmidrule(lr){4-6}
Method & RMSE & Bias & RMSE & Bias & Brier \\
\midrule
\multicolumn{6}{l}{\textit{$n_{\mathrm{train}}=1000$, $n_{\mathrm{node}}=60$}} \\
SSE & 1.586 & \textbf{0.000} & 0.092 & \textbf{0.000} & 0.159 \\
Wilcoxon & 1.354 & -0.002 & 0.080 & \textbf{0.000} & 0.157 \\
Pinball & 1.873 & 0.028 & 0.105 & -0.002 & 0.162 \\
CPM--probit & \textbf{1.351} & 0.009 & \textbf{0.073} & -0.001 & \textbf{0.156} \\
CPM--logistic & 1.368 & 0.025 & 0.075 & -0.001 & \textbf{0.156} \\
\addlinespace[3pt]
\multicolumn{6}{l}{\textit{$n_{\mathrm{train}}=2000$, $n_{\mathrm{node}}=120$}} \\
SSE & 1.097 & -0.003 & 0.063 & \textbf{0.000} & 0.155 \\
Wilcoxon & 0.940 & -0.004 & 0.055 & \textbf{0.000} & 0.154 \\
Pinball & 1.641 & 0.026 & 0.090 & -0.002 & 0.159 \\
CPM--probit & \textbf{0.938} & \textbf{0.002} & \textbf{0.051} & \textbf{0.000} & \textbf{0.153} \\
CPM--logistic & 0.953 & 0.022 & 0.052 & \textbf{0.000} & \textbf{0.153} \\
\addlinespace[3pt]
\multicolumn{6}{l}{\textit{$n_{\mathrm{train}}=5000$, $n_{\mathrm{node}}=300$}} \\
SSE & 0.642 & -0.005 & 0.037 & \textbf{0.000} & \textbf{0.152} \\
Wilcoxon & 0.581 & -0.007 & 0.034 & \textbf{0.000} & \textbf{0.152} \\
Pinball & 1.599 & 0.034 & 0.084 & -0.003 & 0.158 \\
CPM--probit & \textbf{0.580} & \textbf{-0.004} & \textbf{0.031} & \textbf{0.000} & \textbf{0.152} \\
CPM--logistic & 0.594 & 0.015 & 0.033 & \textbf{0.000} & \textbf{0.152} \\
\bottomrule
\end{tabular}
\endgroup
\end{table}

\begin{table}[!htbp]
\centering
\begingroup
\footnotesize
\setlength{\tabcolsep}{4pt}
\singlespacing
\caption{Full single-tree results under Scenario~1.2: distributional prediction metrics and computation time.}
\label{tab:supp_tree_s12_distribution}
\begin{tabular}{lrrrrrrrr}
\toprule
 & \multicolumn{2}{c}{CRPS} & & \multicolumn{2}{c}{80\% PI} & \multicolumn{2}{c}{Non-exceedance} & \\
\cmidrule(lr){2-3}\cmidrule(lr){5-6}\cmidrule(lr){7-8}
Method & Mean & Median & CvM & Coverage & Med.\ width & $\hat p_{0.1}$ & $\hat p_{0.9}$ & Time (s) \\
\midrule
\multicolumn{9}{l}{\textit{$n_{\mathrm{train}}=1000$, $n_{\mathrm{node}}=60$}} \\
SSE & 3.172 & 2.230 & 0.432 & 0.773 & 13.154 & 0.111 & 0.885 & \textbf{0.003} \\
Wilcoxon & 3.140 & 2.207 & 0.446 & 0.773 & \textbf{13.029} & 0.112 & 0.886 & 0.005 \\
Pinball & 3.625 & 2.589 & 0.383 & 0.777 & 13.358 & \textbf{0.110} & 0.887 & 0.022 \\
CPM--probit & \textbf{3.106} & 2.165 & \textbf{0.361} & \textbf{0.780} & 13.121 & 0.111 & \textbf{0.891} & 0.108 \\
CPM--logistic & 3.112 & \textbf{2.163} & 0.534 & 0.772 & 13.089 & 0.115 & 0.887 & 0.108 \\
\addlinespace[3pt]
\multicolumn{9}{l}{\textit{$n_{\mathrm{train}}=2000$, $n_{\mathrm{node}}=120$}} \\
SSE & 3.104 & 2.184 & 0.271 & 0.787 & 13.253 & \textbf{0.106} & 0.893 & \textbf{0.005} \\
Wilcoxon & 3.086 & 2.172 & 0.278 & 0.787 & \textbf{13.181} & \textbf{0.106} & 0.893 & 0.008 \\
Pinball & 3.600 & 2.561 & 0.272 & 0.787 & 13.579 & \textbf{0.106} & 0.893 & 0.033 \\
CPM--probit & \textbf{3.069} & 2.149 & \textbf{0.241} & \textbf{0.790} & 13.285 & \textbf{0.106} & \textbf{0.895} & 0.668 \\
CPM--logistic & 3.073 & \textbf{2.144} & 0.359 & 0.783 & 13.283 & 0.109 & 0.892 & 0.755 \\
\addlinespace[3pt]
\multicolumn{9}{l}{\textit{$n_{\mathrm{train}}=5000$, $n_{\mathrm{node}}=300$}} \\
SSE & 3.057 & 2.157 & 0.205 & 0.795 & \textbf{13.293} & \textbf{0.102} & 0.897 & \textbf{0.007} \\
Wilcoxon & 3.053 & 2.155 & 0.206 & 0.795 & 13.313 & \textbf{0.102} & 0.897 & 0.016 \\
Pinball & 3.570 & 2.537 & 0.223 & 0.795 & 13.907 & 0.103 & \textbf{0.898} & 0.062 \\
CPM--probit & \textbf{3.046} & 2.145 & \textbf{0.196} & \textbf{0.796} & 13.465 & \textbf{0.102} & \textbf{0.898} & 0.165 \\
CPM--logistic & 3.049 & \textbf{2.141} & 0.279 & 0.790 & 13.486 & 0.105 & 0.895 & 0.189 \\
\bottomrule
\end{tabular}
\endgroup
\end{table}

\clearpage
\subsection{Full Forest Results}
\label{sec:supp_forest}

Tables~\ref{tab:supp_forest_s11_quantile}--\ref{tab:supp_forest_s11_distribution} report full WRF and QRF results for Scenario~1.1, and Tables~\ref{tab:supp_forest_s12_quantile}--\ref{tab:supp_forest_s12_distribution} report the corresponding results for Scenario~1.2. Tables~\ref{tab:supp_forest_s2_quantile}--\ref{tab:supp_forest_s2_distribution} report results for Scenario~2, where no threshold-probability metric is included because the main target is the full conditional distribution under a smooth, nonlinear DGP.

\begin{table}[!htbp]
\centering
\begingroup
\footnotesize
\setlength{\tabcolsep}{4pt}
\singlespacing
\caption{Full forest results under Scenario~1.1: conditional quantile RMSE, bias, and pinball loss.}
\label{tab:supp_forest_s11_quantile}
\begin{tabular}{lrrrrrrrrr}
\toprule
 & \multicolumn{3}{c}{$\tau=0.1$} & \multicolumn{3}{c}{$\tau=0.5$} & \multicolumn{3}{c}{$\tau=0.9$} \\
\cmidrule(lr){2-4}\cmidrule(lr){5-7}\cmidrule(lr){8-10}
Method & RMSE & Bias & Pinball & RMSE & Bias & Pinball & RMSE & Bias & Pinball \\
\midrule
\multicolumn{10}{l}{\textit{$n_{\mathrm{train}}=1000$, $n_{\mathrm{node}}=10$}} \\
QRF & \textbf{0.351} & \textbf{-0.128} & \textbf{0.184} & \textbf{0.247} & \textbf{0.004} & \textbf{0.411} & \textbf{0.357} & \textbf{0.152} & \textbf{0.184} \\
WRF & 0.353 & -0.130 & \textbf{0.184} & 0.252 & 0.005 & \textbf{0.411} & 0.360 & 0.153 & \textbf{0.184} \\
\addlinespace[3pt]
\multicolumn{10}{l}{\textit{$n_{\mathrm{train}}=2000$, $n_{\mathrm{node}}=20$}} \\
QRF & \textbf{0.289} & \textbf{-0.116} & \textbf{0.181} & \textbf{0.192} & \textbf{0.005} & \textbf{0.406} & \textbf{0.295} & \textbf{0.140} & \textbf{0.181} \\
WRF & 0.290 & -0.117 & \textbf{0.181} & 0.195 & 0.006 & \textbf{0.406} & 0.297 & 0.141 & \textbf{0.181} \\
\addlinespace[3pt]
\multicolumn{10}{l}{\textit{$n_{\mathrm{train}}=5000$, $n_{\mathrm{node}}=50$}} \\
QRF & 0.240 & \textbf{-0.112} & \textbf{0.179} & \textbf{0.145} & \textbf{0.004} & \textbf{0.403} & \textbf{0.247} & \textbf{0.128} & \textbf{0.180} \\
WRF & \textbf{0.238} & -0.113 & \textbf{0.179} & 0.146 & 0.005 & \textbf{0.403} & 0.250 & 0.129 & \textbf{0.180} \\
\bottomrule
\end{tabular}
\endgroup
\end{table}

\begin{table}[!htbp]
\centering
\begingroup
\footnotesize
\setlength{\tabcolsep}{5pt}
\singlespacing
\caption{Full forest results under Scenario~1.1: conditional mean and threshold-probability accuracy. The threshold is $c=0$.}
\label{tab:supp_forest_s11_meanprob}
\begin{tabular}{lrrrrr}
\toprule
 & \multicolumn{2}{c}{Mean} & \multicolumn{3}{c}{Threshold probability} \\
\cmidrule(lr){2-3}\cmidrule(lr){4-6}
Method & RMSE & Bias & RMSE & Bias & Brier \\
\midrule
\multicolumn{6}{l}{\textit{$n_{\mathrm{train}}=1000$, $n_{\mathrm{node}}=10$}} \\
QRF & \textbf{0.258} & \textbf{0.008} & \textbf{0.084} & \textbf{-0.003} & \textbf{0.157} \\
WRF & 0.261 & \textbf{0.008} & 0.085 & \textbf{-0.003} & 0.158 \\
\addlinespace[3pt]
\multicolumn{6}{l}{\textit{$n_{\mathrm{train}}=2000$, $n_{\mathrm{node}}=20$}} \\
QRF & \textbf{0.211} & \textbf{0.008} & \textbf{0.067} & \textbf{-0.003} & \textbf{0.155} \\
WRF & 0.212 & 0.009 & 0.068 & \textbf{-0.003} & \textbf{0.155} \\
\addlinespace[3pt]
\multicolumn{6}{l}{\textit{$n_{\mathrm{train}}=5000$, $n_{\mathrm{node}}=50$}} \\
QRF & \textbf{0.172} & \textbf{0.006} & \textbf{0.054} & \textbf{-0.002} & \textbf{0.154} \\
WRF & \textbf{0.172} & 0.007 & 0.055 & \textbf{-0.002} & \textbf{0.154} \\
\bottomrule
\end{tabular}
\endgroup
\end{table}

\begin{table}[!htbp]
\centering
\begingroup
\footnotesize
\setlength{\tabcolsep}{4pt}
\singlespacing
\caption{Full forest results under Scenario~1.1: distributional prediction metrics and computation time.}
\label{tab:supp_forest_s11_distribution}
\begin{tabular}{lrrrrrrrr}
\toprule
 & \multicolumn{2}{c}{CRPS} & & \multicolumn{2}{c}{80\% PI} & \multicolumn{2}{c}{Non-exceedance} & \\
\cmidrule(lr){2-3}\cmidrule(lr){5-6}\cmidrule(lr){7-8}
Method & Mean & Median & CvM & Coverage & Med.\ width & $\hat p_{0.1}$ & $\hat p_{0.9}$ & Time (s) \\
\midrule
\multicolumn{9}{l}{\textit{$n_{\mathrm{train}}=1000$, $n_{\mathrm{node}}=10$}} \\
QRF & \textbf{0.584} & \textbf{0.432} & 0.384 & \textbf{0.825} & \textbf{2.819} & \textbf{0.090} & \textbf{0.914} & \textbf{2.925} \\
WRF & 0.585 & \textbf{0.432} & \textbf{0.380} & \textbf{0.825} & 2.820 & 0.089 & \textbf{0.914} & 3.990 \\
\addlinespace[3pt]
\multicolumn{9}{l}{\textit{$n_{\mathrm{train}}=2000$, $n_{\mathrm{node}}=20$}} \\
QRF & \textbf{0.577} & \textbf{0.424} & 0.344 & \textbf{0.828} & \textbf{2.802} & \textbf{0.088} & \textbf{0.915} & \textbf{3.422} \\
WRF & \textbf{0.577} & 0.425 & \textbf{0.343} & \textbf{0.828} & \textbf{2.802} & 0.087 & \textbf{0.915} & 4.998 \\
\addlinespace[3pt]
\multicolumn{9}{l}{\textit{$n_{\mathrm{train}}=5000$, $n_{\mathrm{node}}=50$}} \\
QRF & \textbf{0.572} & \textbf{0.422} & \textbf{0.320} & \textbf{0.831} & \textbf{2.787} & \textbf{0.085} & \textbf{0.916} & \textbf{4.899} \\
WRF & \textbf{0.572} & \textbf{0.422} & 0.322 & \textbf{0.831} & 2.788 & \textbf{0.085} & \textbf{0.916} & 7.954 \\
\bottomrule
\end{tabular}
\endgroup
\end{table}

\begin{table}[!htbp]
\centering
\begingroup
\footnotesize
\setlength{\tabcolsep}{4pt}
\singlespacing
\caption{Full forest results under Scenario~1.2: conditional quantile RMSE, bias, and pinball loss.}
\label{tab:supp_forest_s12_quantile}
\begin{tabular}{lrrrrrrrrr}
\toprule
 & \multicolumn{3}{c}{$\tau=0.1$} & \multicolumn{3}{c}{$\tau=0.5$} & \multicolumn{3}{c}{$\tau=0.9$} \\
\cmidrule(lr){2-4}\cmidrule(lr){5-7}\cmidrule(lr){8-10}
Method & RMSE & Bias & Pinball & RMSE & Bias & Pinball & RMSE & Bias & Pinball \\
\midrule
\multicolumn{10}{l}{\textit{$n_{\mathrm{train}}=1000$, $n_{\mathrm{node}}=10$}} \\
QRF & 1.556 & -0.693 & 0.838 & 1.413 & \textbf{-0.098} & 2.205 & 2.613 & 0.957 & 1.159 \\
WRF & \textbf{1.509} & \textbf{-0.643} & \textbf{0.836} & \textbf{1.401} & -0.103 & \textbf{2.203} & \textbf{2.406} & \textbf{0.885} & \textbf{1.149} \\
\addlinespace[3pt]
\multicolumn{10}{l}{\textit{$n_{\mathrm{train}}=2000$, $n_{\mathrm{node}}=20$}} \\
QRF & 1.339 & -0.642 & 0.829 & 1.111 & \textbf{-0.096} & 2.180 & 2.178 & 0.885 & 1.142 \\
WRF & \textbf{1.282} & \textbf{-0.590} & \textbf{0.826} & \textbf{1.091} & -0.097 & \textbf{2.178} & \textbf{1.940} & \textbf{0.806} & \textbf{1.132} \\
\addlinespace[3pt]
\multicolumn{10}{l}{\textit{$n_{\mathrm{train}}=5000$, $n_{\mathrm{node}}=50$}} \\
QRF & 1.168 & -0.624 & 0.821 & 0.845 & -0.097 & 2.164 & 1.864 & 0.822 & 1.129 \\
WRF & \textbf{1.102} & \textbf{-0.571} & \textbf{0.818} & \textbf{0.818} & \textbf{-0.096} & \textbf{2.161} & \textbf{1.572} & \textbf{0.729} & \textbf{1.119} \\
\bottomrule
\end{tabular}
\endgroup
\end{table}

\begin{table}[!htbp]
\centering
\begingroup
\footnotesize
\setlength{\tabcolsep}{5pt}
\singlespacing
\caption{Full forest results under Scenario~1.2: conditional mean and threshold-probability accuracy. The threshold is $c\approx 14.339$.}
\label{tab:supp_forest_s12_meanprob}
\begin{tabular}{lrrrrr}
\toprule
 & \multicolumn{2}{c}{Mean} & \multicolumn{3}{c}{Threshold probability} \\
\cmidrule(lr){2-3}\cmidrule(lr){4-6}
Method & RMSE & Bias & RMSE & Bias & Brier \\
\midrule
\multicolumn{6}{l}{\textit{$n_{\mathrm{train}}=1000$, $n_{\mathrm{node}}=10$}} \\
QRF & 1.504 & 0.042 & 0.089 & \textbf{-0.003} & \textbf{0.158} \\
WRF & \textbf{1.449} & \textbf{0.039} & \textbf{0.085} & \textbf{-0.003} & \textbf{0.158} \\
\addlinespace[3pt]
\multicolumn{6}{l}{\textit{$n_{\mathrm{train}}=2000$, $n_{\mathrm{node}}=20$}} \\
QRF & 1.247 & 0.038 & 0.073 & -0.004 & 0.156 \\
WRF & \textbf{1.178} & \textbf{0.036} & \textbf{0.068} & \textbf{-0.003} & \textbf{0.155} \\
\addlinespace[3pt]
\multicolumn{6}{l}{\textit{$n_{\mathrm{train}}=5000$, $n_{\mathrm{node}}=50$}} \\
QRF & 1.040 & 0.030 & 0.062 & -0.003 & 0.155 \\
WRF & \textbf{0.953} & \textbf{0.026} & \textbf{0.055} & \textbf{-0.002} & \textbf{0.154} \\
\bottomrule
\end{tabular}
\endgroup
\end{table}

\begin{table}[!htbp]
\centering
\begingroup
\footnotesize
\setlength{\tabcolsep}{4pt}
\singlespacing
\caption{Full forest results under Scenario~1.2: distributional prediction metrics and computation time.}
\label{tab:supp_forest_s12_distribution}
\begin{tabular}{lrrrrrrrr}
\toprule
 & \multicolumn{2}{c}{CRPS} & & \multicolumn{2}{c}{80\% PI} & \multicolumn{2}{c}{Non-exceedance} & \\
\cmidrule(lr){2-3}\cmidrule(lr){5-6}\cmidrule(lr){7-8}
Method & Mean & Median & CvM & Coverage & Med.\ width & $\hat p_{0.1}$ & $\hat p_{0.9}$ & Time (s) \\
\midrule
\multicolumn{9}{l}{\textit{$n_{\mathrm{train}}=1000$, $n_{\mathrm{node}}=10$}} \\
QRF & 3.145 & 2.277 & 0.396 & 0.826 & 14.960 & 0.088 & \textbf{0.914} & \textbf{2.927} \\
WRF & \textbf{3.138} & \textbf{2.269} & \textbf{0.380} & \textbf{0.825} & \textbf{14.916} & \textbf{0.089} & \textbf{0.914} & 3.985 \\
\addlinespace[3pt]
\multicolumn{9}{l}{\textit{$n_{\mathrm{train}}=2000$, $n_{\mathrm{node}}=20$}} \\
QRF & 3.107 & 2.239 & 0.363 & 0.829 & 14.794 & 0.086 & 0.916 & \textbf{3.427} \\
WRF & \textbf{3.100} & \textbf{2.232} & \textbf{0.343} & \textbf{0.828} & \textbf{14.727} & \textbf{0.087} & \textbf{0.915} & 5.025 \\
\addlinespace[3pt]
\multicolumn{9}{l}{\textit{$n_{\mathrm{train}}=5000$, $n_{\mathrm{node}}=50$}} \\
QRF & 3.082 & 2.225 & 0.341 & 0.833 & 14.593 & 0.084 & 0.917 & \textbf{4.837} \\
WRF & \textbf{3.074} & \textbf{2.216} & \textbf{0.322} & \textbf{0.831} & \textbf{14.547} & \textbf{0.085} & \textbf{0.916} & 7.972 \\
\bottomrule
\end{tabular}
\endgroup
\end{table}

\begin{table}[!htbp]
\centering
\begingroup
\footnotesize
\setlength{\tabcolsep}{4pt}
\singlespacing
\caption{Full forest results under Scenario~2: conditional quantile RMSE, bias, and pinball loss.}
\label{tab:supp_forest_s2_quantile}
\begin{tabular}{lrrrrrrrrr}
\toprule
 & \multicolumn{3}{c}{$\tau=0.1$} & \multicolumn{3}{c}{$\tau=0.5$} & \multicolumn{3}{c}{$\tau=0.9$} \\
\cmidrule(lr){2-4}\cmidrule(lr){5-7}\cmidrule(lr){8-10}
Method & RMSE & Bias & Pinball & RMSE & Bias & Pinball & RMSE & Bias & Pinball \\
\midrule
\multicolumn{10}{l}{\textit{$n_{\mathrm{train}}=1000$, $n_{\mathrm{node}}=5$}} \\
QRF & 0.702 & -0.279 & 0.256 & 0.519 & 0.181 & 0.737 & 1.024 & 0.019 & 0.463 \\
WRF & \textbf{0.673} & \textbf{-0.244} & \textbf{0.254} & \textbf{0.509} & \textbf{0.165} & \textbf{0.736} & \textbf{1.011} & \textbf{0.009} & \textbf{0.462} \\
\addlinespace[3pt]
\multicolumn{10}{l}{\textit{$n_{\mathrm{train}}=2000$, $n_{\mathrm{node}}=10$}} \\
QRF & 0.649 & -0.271 & 0.251 & 0.453 & 0.167 & 0.730 & \textbf{0.859} & 0.016 & \textbf{0.455} \\
WRF & \textbf{0.617} & \textbf{-0.237} & \textbf{0.248} & \textbf{0.437} & \textbf{0.150} & \textbf{0.728} & 0.860 & \textbf{0.010} & 0.456 \\
\addlinespace[3pt]
\multicolumn{10}{l}{\textit{$n_{\mathrm{train}}=5000$, $n_{\mathrm{node}}=25$}} \\
QRF & 0.607 & -0.263 & 0.247 & 0.404 & 0.153 & 0.725 & \textbf{0.708} & 0.021 & \textbf{0.449} \\
WRF & \textbf{0.572} & \textbf{-0.231} & \textbf{0.245} & \textbf{0.384} & \textbf{0.136} & \textbf{0.723} & 0.732 & \textbf{0.017} & 0.450 \\
\bottomrule
\end{tabular}
\endgroup
\end{table}

\begin{table}[!htbp]
\centering
\begingroup
\footnotesize
\setlength{\tabcolsep}{5pt}
\singlespacing
\caption{Full forest results under Scenario~2: conditional mean RMSE and bias.}
\label{tab:supp_forest_s2_mean}
\begin{tabular}{lrr}
\toprule
Method & Mean RMSE & Mean bias \\
\midrule
\multicolumn{3}{l}{\textit{$n_{\mathrm{train}}=1000$, $n_{\mathrm{node}}=5$}} \\
QRF & 0.460 & 0.009 \\
WRF & \textbf{0.450} & \textbf{0.008} \\
\addlinespace[3pt]
\multicolumn{3}{l}{\textit{$n_{\mathrm{train}}=2000$, $n_{\mathrm{node}}=10$}} \\
QRF & 0.386 & \textbf{0.006} \\
WRF & \textbf{0.375} & \textbf{0.006} \\
\addlinespace[3pt]
\multicolumn{3}{l}{\textit{$n_{\mathrm{train}}=5000$, $n_{\mathrm{node}}=25$}} \\
QRF & 0.329 & \textbf{0.008} \\
WRF & \textbf{0.322} & \textbf{0.008} \\
\bottomrule
\end{tabular}
\endgroup
\end{table}

\begin{table}[!htbp]
\centering
\begingroup
\footnotesize
\setlength{\tabcolsep}{4pt}
\singlespacing
\caption{Full forest results under Scenario~2: distributional prediction metrics and computation time.}
\label{tab:supp_forest_s2_distribution}
\begin{tabular}{lrrrrrrrr}
\toprule
 & \multicolumn{2}{c}{CRPS} & & \multicolumn{2}{c}{80\% PI} & \multicolumn{2}{c}{Non-exceedance} & \\
\cmidrule(lr){2-3}\cmidrule(lr){5-6}\cmidrule(lr){7-8}
Method & Mean & Median & CvM & Coverage & Med.\ width & $\hat p_{0.1}$ & $\hat p_{0.9}$ & Time (s) \\
\midrule
\multicolumn{9}{l}{\textit{$n_{\mathrm{train}}=1000$, $n_{\mathrm{node}}=5$}} \\
QRF & 1.075 & 0.688 & 0.488 & 0.821 & 4.714 & 0.077 & \textbf{0.898} & \textbf{6.127} \\
WRF & \textbf{1.073} & \textbf{0.684} & \textbf{0.443} & \textbf{0.818} & \textbf{4.677} & \textbf{0.080} & \textbf{0.898} & 7.848 \\
\addlinespace[3pt]
\multicolumn{9}{l}{\textit{$n_{\mathrm{train}}=2000$, $n_{\mathrm{node}}=10$}} \\
QRF & 1.062 & 0.676 & 0.467 & 0.828 & 4.750 & 0.074 & \textbf{0.901} & \textbf{6.668} \\
WRF & \textbf{1.060} & \textbf{0.671} & \textbf{0.419} & \textbf{0.825} & \textbf{4.706} & \textbf{0.076} & \textbf{0.901} & 8.625 \\
\addlinespace[3pt]
\multicolumn{9}{l}{\textit{$n_{\mathrm{train}}=5000$, $n_{\mathrm{node}}=25$}} \\
QRF & 1.053 & 0.669 & 0.448 & 0.833 & 4.785 & 0.071 & \textbf{0.904} & \textbf{7.885} \\
WRF & \textbf{1.051} & \textbf{0.664} & \textbf{0.402} & \textbf{0.830} & \textbf{4.740} & \textbf{0.074} & \textbf{0.904} & 11.385 \\
\bottomrule
\end{tabular}
\endgroup
\end{table}

\subsection{Exploratory CPM-Refined Forest Analysis}
To assess whether the improvement from CPM refinement in single-tree predictions remained effective after forest aggregation, we compared the ECDF-based WRF with a CPM-refined WRF under Scenarios 1.1 and 1.2. This exploratory comparison used the first 100 Monte Carlo replicates from the Scenario~1 forest experiment, with $n_{\mathrm{train}}=2000$, $n_{\mathrm{test}}=1000$, 1000 trees, $\textit{mtry}=2$, a subsampling fraction of 0.632, and $n_{\mathrm{node}}=20$. Within each replicate, the two methods used identical Wilcoxon tree partitions and forest randomization. The ECDF-based WRF used the terminal-node ECDF for each tree, whereas the CPM-refined WRF fitted a probit-link CPM with terminal-node indicators and used its fitted conditional CDF as the tree-level prediction. The resulting tree-level CDFs were averaged to form the forest prediction. Both methods were evaluated on the same test observations using the metrics described in the main simulation study, so the comparison was paired by Monte Carlo replicate.

Tables~\ref{tab:supp_cpmforest_quantile}--\ref{tab:supp_cpmforest_distribution} present the results. The WRF results reported here are based on the same first 100 replicates from the forest simulation described in Supplementary Section~\ref{sec:supp_forest}. In both scenarios, the CPM-refined WRF reduced RMSE for the 0.1 and 0.9 conditional quantiles and produced slightly smaller mean and median CRPS. However, it had larger RMSE for the conditional median and mean, slightly worse threshold-probability RMSE and Brier scores, and a larger CvM statistic. The CRPS reductions were less than 1\%, whereas mean computation time increased from approximately 5 seconds to 282--283 seconds, about 57-fold. Thus, CPM refinement produced mixed changes in predictive accuracy and no consistent overall improvement sufficient to offset its additional computation.

\begin{table}[!htbp]
\centering
\begingroup
\footnotesize
\setlength{\tabcolsep}{4pt}
\singlespacing
\caption{Exploratory comparison of the ECDF-based and CPM-refined WRFs: conditional quantile RMSE, bias, and pinball loss. Results are averages over the first 100 Scenario~1 forest replicates with $n_{\mathrm{train}}=2000$ and $n_{\mathrm{node}}=20$.}
\label{tab:supp_cpmforest_quantile}
\begin{tabular}{lrrrrrrrrr}
\toprule
 & \multicolumn{3}{c}{$\tau=0.1$} & \multicolumn{3}{c}{$\tau=0.5$} & \multicolumn{3}{c}{$\tau=0.9$} \\
\cmidrule(lr){2-4}\cmidrule(lr){5-7}\cmidrule(lr){8-10}
Method & RMSE & Bias & Pinball & RMSE & Bias & Pinball & RMSE & Bias & Pinball \\
\midrule
\multicolumn{10}{l}{\textit{Scenario~1.1}} \\
ECDF-based WRF & 0.293 & \textbf{-0.118} & 0.182 & \textbf{0.195} & 0.004 & \textbf{0.407} & 0.296 & \textbf{0.140} & \textbf{0.181} \\
CPM-refined WRF & \textbf{0.266} & -0.124 & \textbf{0.181} & 0.200 & \textbf{0.002} & \textbf{0.407} & \textbf{0.289} & 0.151 & \textbf{0.181} \\
\addlinespace[3pt]
\multicolumn{10}{l}{\textit{Scenario~1.2}} \\
ECDF-based WRF & 1.292 & -0.597 & 0.829 & \textbf{1.096} & \textbf{-0.103} & \textbf{2.180} & 1.942 & \textbf{0.802} & 1.132 \\
CPM-refined WRF & \textbf{1.182} & \textbf{-0.596} & \textbf{0.825} & 1.134 & -0.129 & 2.181 & \textbf{1.868} & 0.865 & \textbf{1.130} \\
\bottomrule
\end{tabular}
\endgroup
\end{table}

\begin{table}[!htbp]
\centering
\begingroup
\footnotesize
\setlength{\tabcolsep}{5pt}
\singlespacing
\caption{Exploratory comparison of the ECDF-based and CPM-refined WRFs: conditional mean and threshold-probability accuracy. The threshold is $c=0$ in Scenario~1.1 and $c\approx 14.339$ in Scenario~1.2.}
\label{tab:supp_cpmforest_meanprob}
\begin{tabular}{lrrrrr}
\toprule
 & \multicolumn{2}{c}{Mean} & \multicolumn{3}{c}{Threshold probability} \\
\cmidrule(lr){2-3}\cmidrule(lr){4-6}
Method & RMSE & Bias & RMSE & Bias & Brier \\
\midrule
\multicolumn{6}{l}{\textit{Scenario~1.1}} \\
ECDF-based WRF & \textbf{0.213} & \textbf{0.008} & \textbf{0.068} & \textbf{-0.003} & \textbf{0.155} \\
CPM-refined WRF & 0.217 & 0.009 & 0.070 & -0.004 & 0.156 \\
\addlinespace[3pt]
\multicolumn{6}{l}{\textit{Scenario~1.2}} \\
ECDF-based WRF & \textbf{1.184} & 0.032 & \textbf{0.068} & \textbf{-0.003} & \textbf{0.155} \\
CPM-refined WRF & 1.217 & \textbf{0.029} & 0.070 & -0.004 & 0.156 \\
\bottomrule
\end{tabular}
\endgroup
\end{table}

\begin{table}[!htbp]
\centering
\begingroup
\footnotesize
\setlength{\tabcolsep}{4pt}
\singlespacing
\caption{Exploratory comparison of the ECDF-based and CPM-refined WRFs: distributional prediction metrics and computation time.}
\label{tab:supp_cpmforest_distribution}
\begin{tabular}{lrrrrrrrr}
\toprule
 & \multicolumn{2}{c}{CRPS} & & \multicolumn{2}{c}{80\% PI} & \multicolumn{2}{c}{Non-exceedance} & \\
\cmidrule(lr){2-3}\cmidrule(lr){5-6}\cmidrule(lr){7-8}
Method & Mean & Median & CvM & Coverage & Med.\ width & $\hat p_{0.1}$ & $\hat p_{0.9}$ & Time (s) \\
\midrule
\multicolumn{9}{l}{\textit{Scenario~1.1}} \\
ECDF-based WRF & 0.578 & 0.425 & \textbf{0.318} & \textbf{0.828} & \textbf{2.804} & \textbf{0.088} & \textbf{0.916} & \textbf{4.959} \\
CPM-refined WRF & \textbf{0.577} & \textbf{0.424} & 0.363 & 0.833 & 2.808 & 0.085 & 0.917 & 282.100 \\
\addlinespace[3pt]
\multicolumn{9}{l}{\textit{Scenario~1.2}} \\
ECDF-based WRF & 3.102 & 2.230 & \textbf{0.318} & \textbf{0.828} & \textbf{14.740} & \textbf{0.088} & \textbf{0.916} & \textbf{4.987} \\
CPM-refined WRF & \textbf{3.096} & \textbf{2.213} & 0.363 & 0.833 & 14.847 & 0.085 & 0.917 & 283.109 \\
\bottomrule
\end{tabular}
\endgroup
\end{table}

\clearpage
\section{Additional Application Details and Results}
\label{sec:supp_application}

\subsection{Baseline Characteristics}
\label{sec:supp_baseline}

Table~\ref{tab:supp_baseline} summarizes baseline characteristics for the $N=11,070$ participants included in the HIV application.

\begin{table}[!htbp]
\centering
\begingroup
\footnotesize
\singlespacing
\caption{Baseline characteristics of participants with valid 6-month viral load and CD4 measurements after ART initiation.}
\label{tab:supp_baseline}
\begin{tabular}{l r}
\toprule
 & \textbf{Overall ($N=11,070$)} \\
\midrule
\multicolumn{2}{l}{\textbf{Age at ART initiation}} \\
\hspace*{1em}Median [IQR] & 31.9 [26.2, 40.2] \\
\multicolumn{2}{l}{\textbf{Sex at birth}} \\
\hspace*{1em}Male & 9272 (83.8\%) \\
\hspace*{1em}Female & 1798 (16.2\%) \\
\multicolumn{2}{l}{\textbf{Site}} \\
\hspace*{1em}Argentina & 374 (3.4\%) \\
\hspace*{1em}Brazil & 2844 (25.7\%) \\
\hspace*{1em}Chile & 3406 (30.8\%) \\
\hspace*{1em}Honduras & 112 (1.0\%) \\
\hspace*{1em}Mexico & 1038 (9.4\%) \\
\hspace*{1em}Peru & 3296 (29.8\%) \\
\multicolumn{2}{l}{\textbf{Mode of HIV infection}} \\
\hspace*{1em}Heterosexual contact & 3466 (31.3\%) \\
\hspace*{1em}Homosexual contact & 6956 (62.8\%) \\
\hspace*{1em}Other & 118 (1.1\%) \\
\hspace*{1em}Unknown & 530 (4.8\%) \\
\multicolumn{2}{l}{\textbf{Previous AIDS event}} \\
\hspace*{1em}Yes & 2450 (22.1\%) \\
\hspace*{1em}No & 8620 (77.9\%) \\
\multicolumn{2}{l}{\textbf{Baseline CD4 cell count}} \\
\hspace*{1em}Median [IQR] & 279 [119, 447] \\
\multicolumn{2}{l}{\textbf{Baseline viral load}} \\
\hspace*{1em}Median [IQR] & 82,618 [18,100, 303,698] \\
\multicolumn{2}{l}{\textbf{ART regimen}} \\
\hspace*{1em}INSTI-based & 3328 (30.1\%) \\
\hspace*{1em}NNRTI-based & 6528 (59.0\%) \\
\hspace*{1em}PI-based & 1172 (10.6\%) \\
\hspace*{1em}Other & 42 (0.4\%) \\
\multicolumn{2}{l}{\textbf{Months from ART initiation to biomarker measurement}} \\
\hspace*{1em}Median [IQR] & 5.75 [4.37, 6.87] \\
\multicolumn{2}{l}{\textbf{Calendar year of ART initiation}} \\
\hspace*{1em}Median [IQR] & 2016 [2014, 2018] \\
\bottomrule
\end{tabular}
\endgroup
\end{table}

\subsection{Tuning and Cross-Validated Performance}
\label{sec:supp_tuning}

For the 6-month CD4 analysis, Table~\ref{tab:supp_tune_cd4_crps} reports 10-fold cross-validated tuning results for QRF and WRF over minimum node sizes $n_{\mathrm{node}}\in\{5,10,20,50,100,200\}$. The best CRPS was attained at $n_{\mathrm{node}}=5$, which was used in the final CD4 analysis. 

For the viral load analysis, Table~\ref{tab:supp_tune_vl} reports 10-fold cross-validated Brier scores for $P(\mathrm{VL}<80\mid\mathbf{X})$ from WRF over $n_{\mathrm{node}}\in\{20,50,100,200,500,1000\}$. Larger candidate values were used because 83.1\% of participants had viral load below the common detection limit, creating extensive ties. The best Brier score was attained at $n_{\mathrm{node}}=20$.

\begin{table}[!htbp]
\centering
\begingroup
\small
\singlespacing
\caption{CD4 tuning: 10-fold cross-validated mean and median CRPS for QRF and WRF at each candidate minimum node size.}
\label{tab:supp_tune_cd4_crps}
\begin{tabular}{r cc cc}
\toprule
 & \multicolumn{2}{c}{QRF} & \multicolumn{2}{c}{WRF} \\
\cmidrule(lr){2-3}\cmidrule(lr){4-5}
$n_{\mathrm{node}}$ & Mean CRPS & Median CRPS & Mean CRPS & Median CRPS \\
\midrule
5   & \textbf{69.052} & \textbf{48.195} & \textbf{69.516} & \textbf{47.766} \\
10  & 77.177 & 54.089 & 77.419 & 53.556 \\
20  & 82.801 & 57.914 & 83.055 & 57.516 \\
50  & 88.885 & 62.344 & 89.060 & 61.932 \\
100 & 93.285 & 65.678 & 93.361 & 65.008 \\
200 & 98.306 & 69.081 & 98.493 & 68.225 \\
\bottomrule
\end{tabular}
\endgroup
\end{table}

\begin{table}[!htbp]
\centering
\begingroup
\small
\singlespacing
\caption{Viral load tuning: 10-fold cross-validated Brier score for $P(\mathrm{VL}<80\mid\mathbf{X})$ from WRF at each candidate minimum node size.}
\label{tab:supp_tune_vl}
\begin{tabular}{r c}
\toprule
$n_{\mathrm{node}}$ & Brier score \\
\midrule
20   & \textbf{0.116} \\
50   & 0.123 \\
100  & 0.127 \\
200  & 0.129 \\
500  & 0.132 \\
1000 & 0.134 \\
\bottomrule
\end{tabular}
\endgroup
\end{table}

Table~\ref{tab:supp_vl_brier} summarizes the final 10-fold cross-validated Brier scores for viral suppression. The classification random forest directly targets the binary endpoint, whereas WRF estimates this threshold probability from the fitted ordered-outcome distribution.

\begin{table}[!htbp]
\centering
\begingroup
\small
\singlespacing
\caption{Cross-validated Brier scores for predicting viral suppression, $P(\mathrm{VL}<80\mid\mathbf{X})$.}
\label{tab:supp_vl_brier}
\begin{tabular}{l c}
\toprule
Method & Brier score \\
\midrule
Uninformative model & 0.1400 \\
Logistic regression with restricted cubic splines & 0.1270 \\
WRF ordered-outcome model & 0.1160 \\
Classification random forest & \textbf{0.0848} \\
\bottomrule
\end{tabular}
\endgroup
\end{table}

\subsection{CD4 Calibration}
\label{sec:supp_calibration}

Figure~\ref{fig:supp_calib_cd4} shows calibration plots for QRF and WRF predictions of 6-month CD4. For each clinical threshold ($c=200,350,500$ cells/$\mu$L), participants are grouped by predicted probability $\hat P(\mathrm{CD4}\leq c\mid\mathbf{X})$, and the observed proportion below $c$ is plotted against the mean predicted probability. 
The calibration curves for the QRF and WRF are generally similar across the three CD4 thresholds, consistent with the small differences in the CvM and Brier scores.

\begin{figure}[!htbp]
\centering
\IfFileExists{calibration_cd4.pdf}{%
\includegraphics[width=\linewidth]{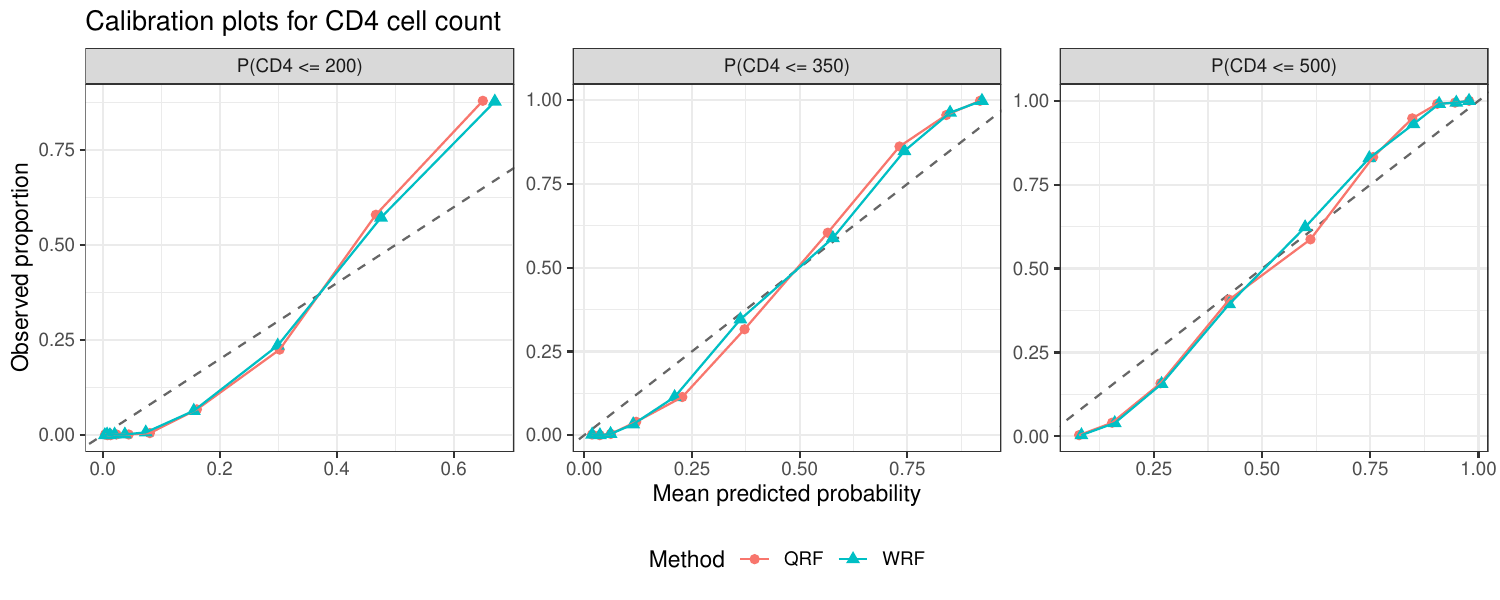}%
}{%
\fbox{\parbox{0.90\linewidth}{\centering Placeholder for \texttt{calibration\_cd4.pdf}.}}%
}
\caption{Calibration plots for 6-month CD4 predictions from QRF and WRF. Each panel plots the observed proportion below a threshold against the mean predicted probability, for thresholds $c=200,350,500$ cells/$\mu$L. The dashed line indicates perfect calibration.}
\label{fig:supp_calib_cd4}
\end{figure}
\end{document}